\documentclass[a4paper]{article}

\usepackage[utf8]{inputenc}
\usepackage[english]{babel} 
\usepackage[T1]{fontenc}

\usepackage{natbib}
\setcitestyle{aysep={,}} 

\usepackage{ upgreek }
\usepackage{gensymb} 
\makeatletter
\newcommand{\ssymbol}[1]{^{\@fnsymbol{#1}}}
\makeatother

\usepackage{amsmath}
\usepackage{graphicx}
\usepackage[colorinlistoftodos]{todonotes}
\usepackage[colorlinks=true, allcolors=blue]{hyperref}

\usepackage{subcaption}
\usepackage{amsmath}
\usepackage{booktabs}
\usepackage{hyperref}

\usepackage{graphicx}%
\usepackage{multirow}%
\usepackage{amsmath,amssymb,amsfonts}%
\usepackage{amsthm}%
\usepackage{mathrsfs}%
\usepackage[title]{appendix}%
\usepackage{xcolor}%
\usepackage{textcomp}%
\usepackage{manyfoot}%
\usepackage{booktabs}%
\usepackage{algorithm}%
\usepackage{algorithmicx}%
\usepackage{algpseudocode}%
\usepackage{listings}%
\usepackage{url}
\usepackage{graphicx}
\usepackage{subcaption}
\usepackage{subcaption} % put this in the preamble
\usepackage{amsmath}
\usepackage{setspace}

\usepackage[a4paper,top=3cm,bottom=2cm,left=3cm,right=3cm,marginparwidth=1.75cm]{geometry}

\usepackage{authblk}

\DeclareMathOperator*{\argmin}{argmin}
\DeclareMathOperator*{\argmax}{argmax}

\newtheorem{theorem}{Theorem}%  meant for continuous numbers
\newtheorem{proposition}[theorem]{Proposition}% 
\newtheorem{remark}{Remark}%

\newtheorem{definition}{Definition}%

\newtheorem{lemma}{Lemma}
\newtheorem{corollary}{Corollary}

\DeclareUnicodeCharacter{2212}{\ensuremath{-}}

\title{Local depth-based classification of directional data}

\author[$\ssymbol{1}$]{Giuseppe Gismondi}
\author[$\ssymbol{2}$]{Rebecca Rivieccio}
\author[$\ssymbol{1}$]{Giuseppe Pandolfo}
\affil[$\ssymbol{1}$]{Dept. of Economics and Statistics, University of Naples Federico II, Napoli, Italy}
\affil[$\ssymbol{2}$]{Dept. of Physics, University of Naples Federico II, Napoli, Italy}
\date{}

\begin{document}

  \maketitle

\begin{abstract}
Directional data arise in many applications where observations are naturally represented as unit vectors or as observations on the surface of a unit hypersphere. In this context, statistical depth functions provide a center--outward ordering of the data. This work aims at proposing the use of a local notion of data depth function to be applied in the DD-plot (Depth vs. Depth plot) to classify directional data. The proposed method is investigated through an extensive simulation study and two real-data examples.
\end{abstract}

\section{Introduction}
\label{sec1}

Directional data analysis is a branch of statistics that is concerned with the exploration and modelling of data expressed as angles or unit vectors.
Such data lie on the surface of the unit hypersphere $S^{q-1}:= \{x \in R^q: \|x\|_2 = 1\} $ of $q-1$ dimensions, where $||x||_2:= \sqrt{\sum_{i=1}^{q}x_i^2}$ and $x= (x_1, ..., x_q)$ .
They naturally arise in a variety of real-world contexts where the vectors represent directions, rotations, or cyclic phenomena. Prominent applications can be found in geology, where the orientation of magnetic fields in rocks is studied, as well as in meteorology and psychology, in the analysis of wind directions or the perception of the spatial orientation. Further examples and theoretical developments are discussed by \cite{mardia1999directional}, which remains a fundamental reference in the field of directional statistics, later complemented by the more recent work of  \cite{ley2017modern}.
There are several other application domains for which the orientation of vectors in the space contains richer information than the magnitude, such as compositional data (i.e. when vectors consist of nonnegative components that sum up to one) such as the relative frequencies of words in a document \citep[see][]{pandolfo2021depth}.
As noted by \cite{stephens1982use}, applying a square-root transformation to each vector maps these compositions to directional data lying on the surface of a $(q-1)$-dimensional unit hypersphere.

When dealing with this type of data, several challenges arise due to the absence of a natural reference direction and the lack of a unique definition of orientation or sense of rotation.
Moreover, since directional data do not have a natural ordering, the development of suitable depth functions can be quite useful.
Indeed, depth functions provide a notion of centrality, allowing a center--outward ordering of locations on the manifold \citep{agostinelli2013nonparametric} by generalizing the univariate notions of median and rank to the multivariate setting. 

Several notions of depth for directional data have been proposed and employed as feature spaces for implementing supervised classification methods
\citep{pandolfo2021depth,dey2025classification}.
Traditional \textit{global} angular depth functions aim to describe the overall centrality of a point with respect to the entire data distribution, providing a single measure of how central or peripheral an observation is.
However, this approach is only reliable when dealing with unimodal and convexly distributed data.
In the case of multimodality or non-convex data structures, typically arising in mixture models or clustering problems, the global depths fail to provide a meaningful representation of centrality, since multiple local centers may exist \citep{paindaveine2013depth}.

To overcome this limitation, local depth functions have been proposed to provide a more refined assessment of centrality.
Specifically, they aim to evaluate the position of a point within a restricted neighborhood of the data. This way it is possible to capture the centrality at a specific scale of locality \citep{agostinelli2011local}.
Such approach allows for a more flexible characterization of complex data structures, making local depths particularly effective also for classification tasks.

Hence, the goal of this work is to define a local version of the cosine distance depth (CDD) proposed by \cite{pandolfo2018distance} to be exploited for directional data classification purposes. More specifically, we consider its application to the two-step procedure known as DD-plot (Depth vs. Depth plot) introduced by \cite{Liu99} and later used to perform classification by \cite{Li01062012} (i.e. the DD-classifier). Roughly speaking, for two given samples, the corresponding DD-plot represents the depth values of those sample points with respect to the two underlying distributions, and thus transforms the samples in any dimension to a simple two-dimensional scatter plot. Then, a curve that best separates the two samples in their DD-plot is applied, in the sense that the separation yields the smallest classification error in the DD space.

The remainder of the paper is organized as follows. Sections \ref{sec3} and \ref{classDD} briefly recall the concept of data depth for directional data and the classification via the Depth vs. Depth (DD)-plot. Section \ref{theory} introduces a notion local cosine distance depth then used to built a classifier introduced in Section \ref{sec:localDD}. Section \ref{simulation} presents an extensive simulation study to investigate the performances of the proposed method. Section \ref{applications} provides two real-data examples. Finally, some concluding remarks are offered in Section \ref{concl}.

\section{Data depths for directional data}
\label{sec3}

Statistical depth functions extend univariate ordering to higher dimensions. Particularly, they offer a center-outward ordering by providing a measure of how central a point is with respect to a certain distribution. The concept of data depth was first extended to the analysis of directional data by \cite{small87} and later by \cite{liu92}. Accordingly, directional depth functions measure the degree of centrality of a point with respect to a directional distribution, and they provide a center-outward ordering on circles or on hyperspheres. 

Within the literature we can find the angular tukey depth (ATD) and the angular simplicial depth (ASD), which represent the directional extensions of the Tukey’s halfspace depth \citep{tukey1975mathematics} and the simplicial depth originally introduced for data in $\mathbb{R}^{q}$, respectively. Such depths are also known as geometric depths because they are based on geometric stuctures (i.e. hemispheres and simplices). Because of that their main drawback is related to their high computational cost, which makes them unfeasible when $q > 3$.
For such computational issue, here we focus on the class of distance-based depth functions introduced by \cite{pandolfo2018distance}, which includes the arc distance depth (ADD) of \cite{liu92}, the chord distance depth (ChDD) and the cosine distance depth (CDD).
Such distance-based depths are computationally feasible even in high dimensions and are strictly positive everywhere on $S^{q-1}$, except for in the uninteresting case of a point mass distribution, whereas ASD and ATD may take zero values, which can cause issues in supervised classification. In addition, they do not produce ties in the sample case. The computational advantage of CDD stems from the fact that it requires only pairwise inner products $\langle x_i, x_j \rangle$, which can be computed efficiently even in high dimensions. This contrasts with geometric depths that require solving complex optimization problems over hemispheres or simplices.

One more notion of depth for directional data is the angular Mahalanobis depth, which was studied by \cite{ley2014new} and developed by using the concept of directional quantiles. However, its application is often limited by the necessary prior choice of a spherical locational functional.

Here we focus on the CDD because of its computational ease and its properties that are particularly useful in defining our proposal. In the following, we recall its definition.
\begin{definition}[Cosine Distance Depth]
\label{def:defcosdepth}
The cosine distance depth of a point $x \in S^{q-1}$ with respect to the distribution $F$ on $S^{q-1}$ is defined as:
\[
\mathrm{CDD}(x, F) := 2 - \mathbb{E}_F[d_{\cos}(x, W)],
\]
where $d_{\cos}(x, w) = 1 - \langle x, w \rangle$ is the cosine distance, and $W \sim F$.
\end{definition}

The sample version is obtained by replacing $F$ by its empirical counterpart $\hat{F}_{n}$ calculated from the sample $x_1, \ldots, x_n$. The CDD satisfies all the following properties:

\vspace{1em}
\noindent\textbf{P1. Rotation invariance:}
$\mathrm{CDD}(x,F) = \mathrm{CDD}(Ox,OF)$ for any $q \times q$ orthogonal matrix $O$.

\vspace{1em}
\noindent\textbf{P2. Maximality at center:} $\displaystyle \max_{x \in S^{q-1}} \mathrm{CDD}(x,F) = \mathrm{CDD}(x_0,F)$ for any $F$ with center at $x_0$.

\vspace{1em}
\noindent\textbf{P3. Monotonicity on rays from the deepest point:}
$\mathrm{CDD}(\cdot,\cdot)$ decreases along any geodesic path $t \mapsto x_t$ from the deepest point $x_0$ to its antipodal point $-x_0$.

\vspace{1em}
\noindent\textbf{P4. Minimality at the antipodal point to the center:}
 $\mathrm{CDD}(-x_0,F) = \displaystyle \inf_{x \in S^{q-1}} \mathrm{CDD}(x,F)$ 
\quad for any $F$ with center at $x_0$. \\

While CDD provides a robust global measure of centrality, it may not adequately capture local structure in complex data distributions. In the next section, we address this issue by introducing a local version of CDD that adapts to the scale of locality.

\section{Classification in the depth space}
\label{classDD}

After the first suggestion in \cite{liu92}, the use of data depth to perform supervised classification has been suggested and investigated by many authors. Two main approaches have been proposed in the literature: \textit{(i)} the maximum depth classifier and \textit{(ii)} the Depth vs. Depth (DD)-classifier.

The first simply assigns a point $x$ to the distribution (or group) with respect to it attains the highest depth value for any considered depth $\mathrm{D}(\cdot,\cdot)$, that is:
$$
\mathrm{D}(x, \hat{F}_i) > \mathrm{D}(x, \hat{F}_j) \quad i \neq j, \quad \Rightarrow \quad  \textrm{assign}~x~\textrm{to}~\hat{F}_i.
$$
where $\mathrm{D}(x, \hat{F_i})$ and $\mathrm{D}(x, \hat{F_j})$ are the empirical depths of $x$ w.r.t. the $i$-th and the $j$-th distribution, respectively.

The latter was proposed by \cite{Li01062012} and is a refinement of the maximum depth classifier. It is based on the DD-plot (Depth vs. Depth plot), introduced by \cite{Liu99} which is a two-dimensional scatterplot where each data point is represented with coordinates given by its depth evaluated with respect to two distributions. Then some classification rule $s(\cdot)$ is directly applied in the DD-space
\[
C_{s}(x) =
\begin{cases}
\mathrm{D}(x, \hat{F}_i) > s\bigl(\mathrm{D}(x, \hat{F}_j)\bigr) &\Rightarrow \text{assign } x \text{ to } \hat{F}_i,\\[4pt]
\mathrm{D}(x, \hat{F}_i) \leq s\bigl(\mathrm{D}(x, \hat{F}_j)\bigr) &\Rightarrow \text{assign } x \text{ to } \hat{F}_j.
\end{cases}
\]
where $s(\cdot)$ is a real increasing function. \cite{Li01062012} suggested to look for a polynomial separator that is chosen in order to minimize the empirical misclassification error rate on the training sample. 

Note that when $s(\cdot)$ is the identity function, the classification rule becomes the maximum depth classifier.

The maximum depth classifier is certainly quite intuitive and easy to implement. In addition, it can deal with classification problems involving a large number of groups. Conversely, the DD-classifier is more flexible, but it requires the degree of the polynomial function for which the misclassification rate is minimized to be searched for.

The procedure can be applied to any kind of data, providing that a corresponding depth function exists. For instance, DD-plot for functional data have been developed \citep{cuesta2017dd} and also to directional data by means of standard global depths \citep{pandolfo2021depth,demni2019cosine}.

Obviously, an important aspect to be considered regards the choice of the depth function. Indeed, from a classification  perspective, it must be noted that  the halfspace and simplicial depths assign zero depth value to all those points which does not belong to the convex hull of the support of distribution. This implies that sample points lying outside the convex hull of the training set have zero empirical depth, thus it is not possible to assign an observation to one of the competing groups. On the contrary, this does not occur by adopting distance-based depths since they are always positive for any data point in the sample space. 

While existing works have applied DD-classifiers with \textit{global} depth functions, this approach may not capture local structure in complex directional distributions. In the next sections, we introduce a \textit{local} version of the cosine distance depth and develop a corresponding DD-classifier that adapts to the scale of locality.

\section{Local Cosine Distance Depth}
\label{theory} 

To develop a notion of depth capable of describing local features and mode(s) in directional distributions, \cite{agostinelli2012depthSIS} proposed a local version of the ASD by constraining the size of the spherical simplices, while later on \cite{pandolfo2022gld} proposed a local extention of the distance-based depths by restricting a global distance-based depth measure by considering only the points within a certain distance to a given point. 

Drawing inspiration from the work of \cite{paindaveine2013depth}, we propose a local version of the cosine distance depth. This is derived by calculating the depth with respect to the empirical distribution associated with the sample obtained by adding the reflections of the original observations with respect to a given point $x$ to the original sample.

We provide the definition of this depth below and briefly discuss why this approach can be problematic or even fail altogether with other measures of angular depth.

To do that, we need first to define a symmetric reflector operator on the unit hypersphere. Given a point $x_i\in S^{q-1}$, another point $x_j\in{S}^{q-1}$  can be reflected symmetrically through $x_i$ as follows:
\[
R(x_j,x_i)=2x_i\langle x_i\,, x_j\rangle-x_j, 
\]
where $\langle\,\cdot ,\cdot\,\rangle$ is the scalar product between the two vectors. This operator has the following properties:
\begin{enumerate}
    \item $R(-x_i,x_i)=-x_i$,
    \vspace{1em}
    \item $R(x_j,x_i)=R(x_j,-x_i)$,
    \vspace{1em}
    \item Given a distance function $d(\cdot,\cdot)$ defined on the unit hypersphere:
    \[
    d(x_i,R(x_j,x_i))=d(x_i,x_j).
    \]
\end{enumerate}

Consider a sample $X$ on ${S}^{q-1}$ and any given point $x_i \in X$. Let $X_{-i} := X \setminus \{x_i\}$ denote the sample without $x_i$. The set of all such reflected points is denoted by $R_i := \{ R(x_j, x_i) : x_j \in X_{-i} \}$, and the augmented sample is $X^{R_i} := X \cup R_i$. 
%which is the set containing all original points in $X$ and all the reflected points of the dataset through $x_i$. 
Here we assume that $|X|=n$, so $|X_{-i}|=n-1$ and $|X^{R_i}|=2n-1$. $x_i$ should be the depth median of its own reflected region, however this is not always true. 

When considering CDD, it is possible to identify a condition that precisely defines when this occurs. 

%new version
\begin{proposition}
\label{equiv}
Given a point $x_i\in X\subseteq S^{q-1}$ and the reflected region $X^{R_i}$, we have:
\begin{align*}
    x_i&=\argmax_{x_j\in X}\mathrm{CDD}\!\left(x_j,X^{R_i}\right) 
    \quad \text{if} \quad 1+2\sum_{k \neq i}\langle x_i,x_k\rangle>0\\ 
    \argmax_{x_k\in X_{-i}}d_{\cos}(x_i,x_k)&=\argmax_{x_j\in X} \mathrm{CDD}\!\left(x_j,X^{R_i}\right) 
    \quad \text{if} \quad 1+2\sum_{k \neq i}\langle x_i,x_k\rangle<0\\
    X&=\argmax_{x_j\in X}\mathrm{CDD}\!\left(x_j,X^{R_i}\right)
    \quad \text{if} \quad 1+2\sum_{k \neq i}\langle x_i,x_k\rangle=0
\end{align*}
where $d_{\cos}(x,y) = 1 - \langle x, y \rangle$.

Hence, $x_i$ is either a depth median or antipodal to the depth median of the region.
\end{proposition}

\begin{proof}
Let $X = \{x_1,\dots,x_n\} \subset S^{q-1}$. For a fixed $x_i$, define the reflected set through $x_i$:
\[
X^{R_i} = \big\{ x_k^{R_i} : x_k^{R_i} = 2x_i\langle x_i,x_k\rangle - x_k,\; k \neq i \big\}.
\]

The CDD of a point $x_j$ with respect to $(X, X^{R_i})$ is
\[
\mathrm{CDD}\!\left(x_j, X^{R_i}\right) =
2 - \frac{1}{2(n-1)}\!\left[
\sum_{k \ne j} d_{\cos}(x_j,x_k)
+ \sum_{k \ne i} d_{\cos}\!\left(x_j, x_k^{R_i}\right)
\right].
\]

Maximizing CDD is equivalent to minimizing
\[
f(x_j) = \sum_{k \ne j} d_{\cos}(x_j,x_k)
+ \sum_{k \ne i} d_{\cos}\!\left(x_j, x_k^{R_i}\right).
\]

%\noindent\textbf{1. Computing the two sums separately.}\\
For the first sum, since $d_{\cos}(x_j,x_j)=0$,
\[
\sum_{k \ne j} d_{\cos}(x_j,x_k)
= \sum_{k=1}^n [1 - \langle x_j, x_k \rangle]
= n - \langle x_j, S \rangle,
\]
where $S = \sum_{k=1}^n x_k$.

For the second sum, using
\[
\langle x_j, x_k^{R_i} \rangle
= 2\langle x_i, x_k\rangle \langle x_j, x_i\rangle - \langle x_j, x_k \rangle,
\]
we obtain
\[
d_{\cos}\!\left(x_j, x_k^{R_i}\right)
= 1 - 2\langle x_i, x_k\rangle \langle x_j, x_i\rangle + \langle x_j, x_k\rangle.
\]

Summing over $k \ne i$:
\[
\sum_{k \ne i} d_{\cos}\!\left(x_j, x_k^{R_i}\right)
= (n-1)
- 2\langle x_j, x_i\rangle \sum_{k \ne i} \langle x_i, x_k\rangle
+ \langle x_j, \sum_{k \ne i} x_k \rangle.
\]

%\noindent\textbf{2. Combining the sums.}\\
Thus
\[
f(x_j) = \big[ n - \langle x_j, S \rangle \big]
+ \big[ (n-1) - 2\langle x_j, x_i\rangle A + \langle x_j, \sum_{k \ne i} x_k \rangle \big],
\]
where $A = \sum_{k \ne i} \langle x_i, x_k\rangle$.

Since $S = x_i + \sum_{k \ne i} x_k$, we have
\[
-\langle x_j, S \rangle + \langle x_j, \sum_{k \ne i} x_k \rangle
= -\langle x_j, x_i\rangle.
\]

Therefore
\[
f(x_j) = 2n - 1 - \langle x_j, x_i\rangle - 2\langle x_j, x_i\rangle A.
\]

Let $v = \langle x_j, x_i\rangle$, and $g(v)=f(x_j)$. Then
\[
g(v) = 2n - 1 - v(1 + 2A).
\]

Since $g$ is linear in $v$, and $v \in [-1,1]$ for unit vectors, the minimizer depends on the sign of $C = 1 + 2A$:

\begin{itemize}
    \item If $C > 0$: $g$ is decreasing in $v$. \\
    Minimum occurs at the largest possible $v$ in $X$, which is $v = 1$ attained at $x_j = x_i$. \\
    Hence $x_i$ is the unique minimizer of $f$ $\Rightarrow$ $x_i$ maximizes CDD.

    \item If $C < 0$: $g$ is increasing in $v$. \\
    Minimum occurs at the smallest possible $v$ in $X$. \\
    The smallest $v$ is $-1$ if $-x_i \in X$; otherwise it is $\min_{x_j \in X} \langle x_j, x_i\rangle$, 
    which corresponds to $\max_{x_j \in X} d_{\cos}(x_i, x_j)$. \\
    Hence the maximizer of $d_{\cos}(x_i, \cdot)$ in $X_{-i}$ maximizes CDD.

    \item If $C = 0$: $g$ is constant w.r.t. $v$ $\Rightarrow$ all $x_j \in X$ give the same CDD $\Rightarrow$ 
    every point in $X$ is a maximizer.
\end{itemize}

In the case $C < 0$, the farthest point from $x_i$ in cosine distance is closest to the antipode $-x_i$. 
Thus, $x_i$ is either the depth median of $X^{R_i}$ (when $C > 0$), or its antipode (when $C < 0$).
\end{proof}

This result has an intuitive interpretation: when $x_i$ is centrally located relative to other points (positive $C$), it becomes the depth median of its reflected region. When $x_i$ is peripheral (negative $C$), its antipode becomes more central in the reflected region.
Note that it may happen that a point $x_i$ is not the depth median of its own reflected region. In particular, when a point $x_i$ is antipodal to the depth median of $X^{R_i}$, then the points should be reordered in an increasing order of depth. In such case, points with lower depth are more similar to $x_i$ according to property \textbf{P3}.

Now, we can define the depth-based neighbourhood of a certain point $x_i$ at level $\beta$ by using any angular depth measures $\mathrm{AD}(\cdot,\cdot)$.

\begin{definition}
\label{beta_neigh}
\textit{($\beta$-Depth based neighbourhood)}. The $\beta$-Depth neighbourhood of a point $x_i \in X \subseteq S^{q-1}$ noted $DN_i^{(\beta)}$, with $\beta\in(0,1]$, is defined as the set of the first $\beta(n-1) $  points in $X_{-i}$ reordered in the following way:

\begin{align*}
\mathrm{AD}\Big(x_j,X^{R_i}\Big)>\mathrm{AD}\Big(x_k,X^{R_i}\Big) > \ldots > \mathrm{AD}\Big(x_n,X^{R_i}\Big)\quad&\text{if}\quad x_i=\argmax_{x_j}\mathrm{AD}\Big(x_j,X^{R_i}\Big);\\
\mathrm{AD}\Big(x_j,X^{R_i}\Big) < \mathrm{AD}\Big(x_k,X^{R_i}\Big) < \ldots < \mathrm{AD}\Big(x_n,X^{R_i}\Big)\quad&\text{if}\quad x_i=\argmin_{x_j}\mathrm{AD}\Big(x_j,X^{R_i}\Big).  
\end{align*}
\end{definition}

\begin{proposition}
\label{prop2}
Consider a point $x_i\in X\subseteq S^{q-1}$ and the $\beta$-depth based neighbourhood $DN^{(\beta)}_i$, and let $CN_i$ denote the set of the first $\beta(n-1)$ points in $X_{-i}$ reordered by increasing cosine distance from $x_i$:
\[
d_{\cos}(x_j,x_i) < d_{\cos}(x_k,x_i) < \ldots < d_{\cos}(x_n,x_i),
\]
then if the CDD is adopted as the angular depth measure, we have
\[
DN^{(\beta)}_i = CN_i.
\]
\end{proposition}

\begin{proof}
From Proposition \ref{equiv}, we know there are three cases depending on $C = 1 + 2\sum_{k\neq i}\langle x_i,x_k\rangle$:\\

\noindent \textbf{Case 1: $C > 0$} \\
Here $x_i = \underset{x_k}{\argmax}~\mathrm{CDD}(x_k,X^{R_i})$, so by Definition \ref{beta_neigh}, points are ordered by decreasing CDD. We need to show:
\[
\mathrm{CDD}(x_j,X^{R_i}) > \mathrm{CDD}(x_l,X^{R_i}) \iff d_{\cos}(x_i,x_j) < d_{\cos}(x_i,x_l).
\]

From the proof of Proposition \ref{equiv}, we have:
\[
f(x_j) = 2n - 1 - \langle x_i, x_j\rangle(1 + 2A),
\]
where $A = \sum_{k\neq i} \langle x_i, x_k\rangle$. Since maximizing CDD is equivalent to minimizing $f$, we have:
\[
\mathrm{CDD}(x_j,X^{R_i}) > \mathrm{CDD}(x_l,X^{R_i}) \iff f(x_j) < f(x_l).
\]

Now,
\[
f(x_l) - f(x_j) = -[\langle x_i, x_l\rangle - \langle x_i, x_j\rangle](1 + 2A).
\]

Since $C = 1 + 2A > 0$ in this case:
\[
f(x_l) - f(x_j) > 0 \iff \langle x_i, x_l\rangle - \langle x_i, x_j\rangle < 0 \iff \langle x_i, x_j\rangle > \langle x_i, x_l\rangle.
\]

Converting to cosine distances:
\[
\langle x_i, x_j\rangle > \langle x_i, x_l\rangle \iff 1 - d_{\cos}(x_i,x_j) > 1 - d_{\cos}(x_i,x_l) \iff d_{\cos}(x_i,x_j) < d_{\cos}(x_i,x_l).
\]

Thus, in Case 1, ordering by decreasing CDD is equivalent to ordering by increasing cosine distance from $x_i$.\\

\noindent \textbf{Case 2: $C < 0$} \\
Here $x_i = \underset{x_k}{\argmin}~\mathrm{CDD}(x_k,X^{R_i})$, so by Definition \ref{beta_neigh}, points are ordered by increasing CDD. We need to show:
\[
\mathrm{CDD}(x_j,X^{R_i}) < \mathrm{CDD}(x_l,X^{R_i}) \iff d_{\cos}(x_i,x_j) < d_{\cos}(x_i,x_l).
\]

Since $f(x_j) = 2n-1 - \langle x_i, x_j\rangle(1+2A)$ and $C = 1+2A < 0$, $f$ is increasing in $v = \langle x_i, x_j\rangle$. Therefore:
\[
f(x_j) < f(x_l) \iff \langle x_i, x_j\rangle < \langle x_i, x_l\rangle.
\]

Recall that maximizing CDD is equivalent to minimizing $f$, so:
\[
\mathrm{CDD}(x_j,X^{R_i}) > \mathrm{CDD}(x_l,X^{R_i}) \iff f(x_j) < f(x_l) \iff \langle x_i, x_j\rangle < \langle x_i, x_l\rangle.
\]

Taking the contrapositive:
\[
\mathrm{CDD}(x_j,X^{R_i}) < \mathrm{CDD}(x_l,X^{R_i}) \iff \langle x_i, x_j\rangle > \langle x_i, x_l\rangle.
\]

Converting to cosine distances:
\[
\langle x_i, x_j\rangle > \langle x_i, x_l\rangle \iff 1 - d_{\cos}(x_i,x_j) > 1 - d_{\cos}(x_i,x_l) \iff d_{\cos}(x_i,x_j) < d_{\cos}(x_i,x_l).
\]

Thus, in Case 2, ordering by increasing CDD is also equivalent to ordering by increasing cosine distance from $x_i$.\\

\noindent \textbf{Case 3: $C = 0$} \\
Here all points have equal CDD, so any ordering yields the same set $DN^{(\beta)}_i$, which equals $CN_i$ trivially.

Therefore, in all cases, $DN^{(\beta)}_i = CN_i$.
\end{proof}

Hence, defining the $\beta$-depth neighbourhood of a point $x_i$ is equivalent to looking for the $\beta(n-1)$ nearest neighbours of $x_i$. This result was already encountered by \cite{paindaveine2013depth} for depths in $\mathbb{R}^{q}$ for $q=1$, while here it holds for $q \geq 1$ as long as the CDD is considered. This allows us to compute the reflected region and the depth of a point by looking for the nearest points. This way, the local cosine distance-depth can be defined as follows.

\begin{definition}\label{def:lcdd_sample}\textit{(Local cosine distance depth}). The local cosine distance depth of a point $x_i\in X\subseteq S^{q-1}$ at a locality level $\beta\in(0,1]$ is defined as follows:
\[
\mathrm{LCDD}^{(\beta)}(x_i,X)=\mathrm{CDD}\Big(x_i,DN_i^{(\beta)}\Big).
\]
\end{definition} 

When $\beta$ is set equal to one, the global cosine distance depth is obtained. 
It is worth noting that the chord or the arc distance depths do not guarantee that a point $x_i$ is a depth median or antipodal to the depth median of $X^{R_i}$. This makes reordering the points according to depth in a given region difficult, if possible at all. This occurs because chord and arc distances lack the linear structure that makes the minimization problem in Proposition~\ref{equiv} tractable. Specifically, the function analogous to $f(x_j)$ in the proof does not simplify to a linear function of $\langle x_i, x_j\rangle$ for these distances, making it impossible to derive a simple condition for when $x_i$ is the depth median of $X^{R_i}$.

Note that since cosine distance and neighbourhoods are rotation-invariant, so too are nearest neighbours and local cosine distance depth, i.e. $\mathrm{LCDD}^{(\beta)}(O x_i, O X) = \mathrm{LCDD}^{(\beta)}(x_i, X)$. LCDD is also bounded and strictly positive, i.e. $0 < \mathrm{LCDD}^{(\beta)}(x_i, X) \leq 2$ for a non-point mass $X$, since neighbourhoods are finite and the CDD is strictly positive on $S^{q−1}$.

LCDD does not satisfy the monotonicity property (\textbf{P3}). This is expected and desirable: it is designed to capture \emph{local} centrality, and may exhibit multiple local maxima in multimodal distributions. So, it is monotone in locality. As $\beta \to 1$, LCDD$^{(\beta)}(x) \to$ CDD$(x)$, recovering the global depth.

Figure \ref{fig:fig1} depicts the behaviour of the proposed local cosine distance compared with its global version in the case of a trimodal spherical distribution, showing their contour plots (for $\beta = 0.25$) alongside the density contours of the data. For the sake of illustration, the plots are two-dimensional and the data are reported in spherical coordinates. As one can see, the CDD fail to capture the three modes, identifying only a single global center. In contrast, the proposed LCDD (with $\beta=0.25$) clearly identifies all the three modes.

\begin{figure}[htbp]
\centering

% Top image (a)
\begin{subfigure}{0.51\textwidth}
\centering
\includegraphics[width=\textwidth]{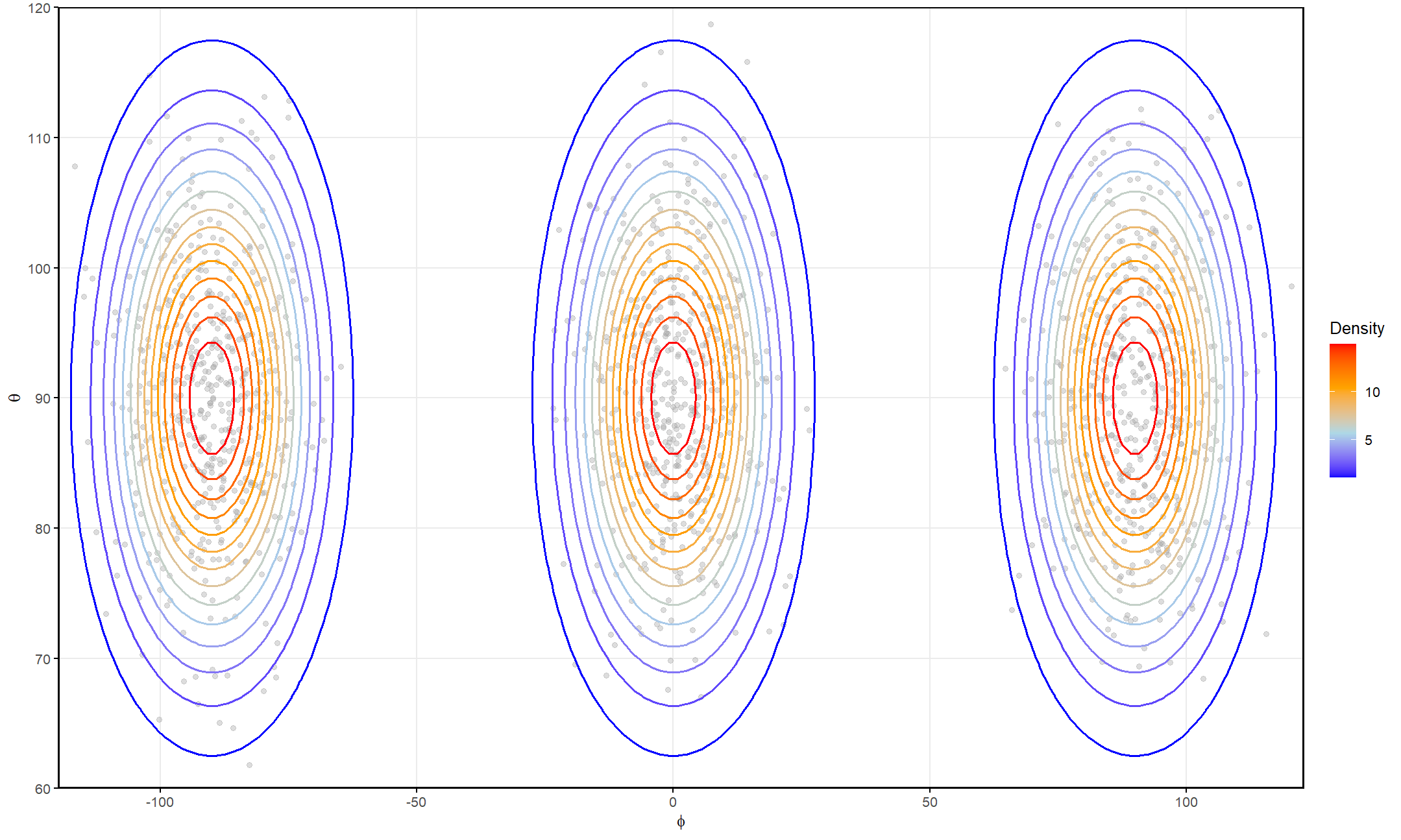}
\caption{}
\end{subfigure}

\vspace{0.3cm}

% Bottom left image (b)
\begin{subfigure}{0.48\textwidth}
\centering
\includegraphics[width=\textwidth]{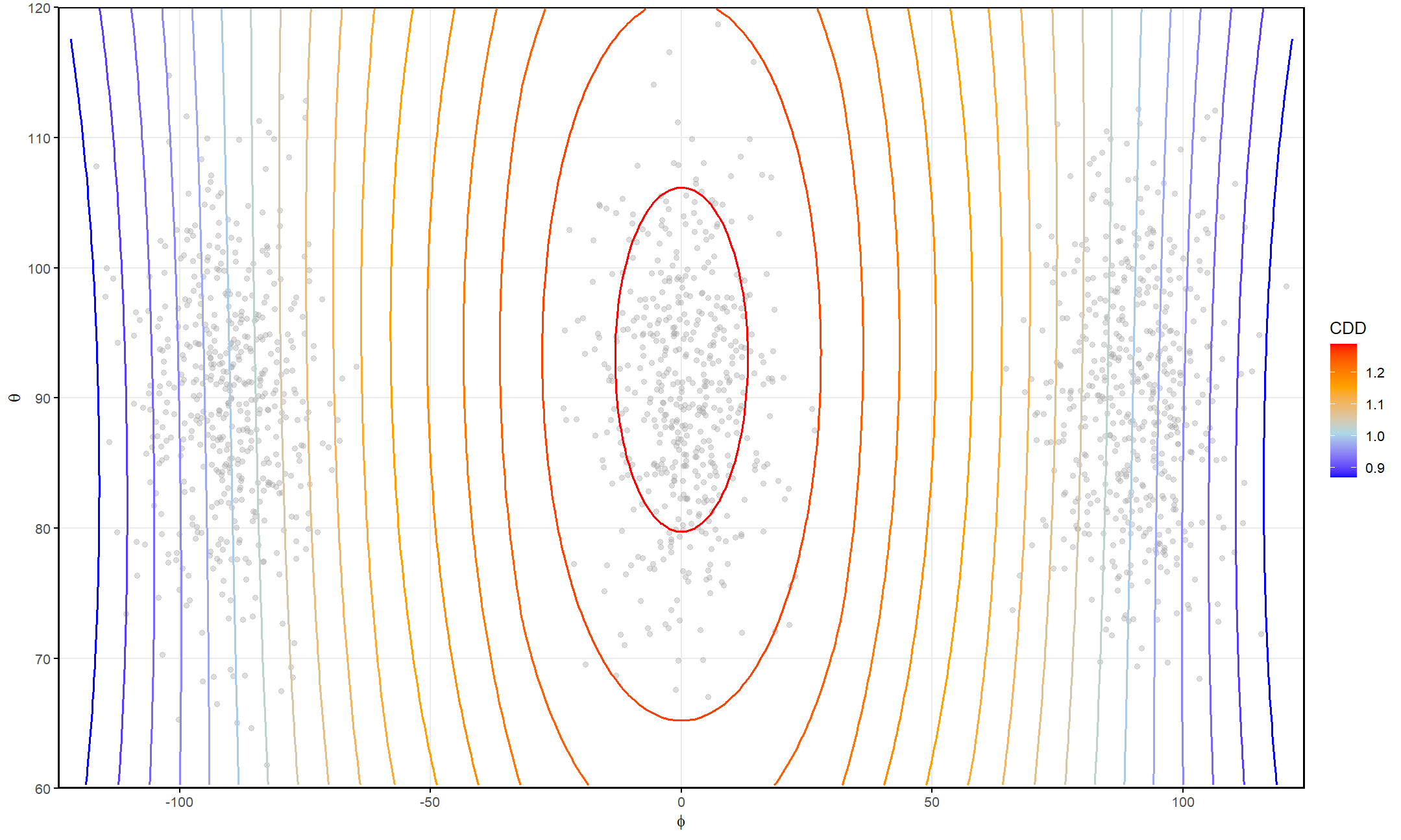}
\caption{}
\end{subfigure}
\hfill
% Bottom right image (c)
\begin{subfigure}{0.48\textwidth}
\centering
\includegraphics[width=\textwidth]{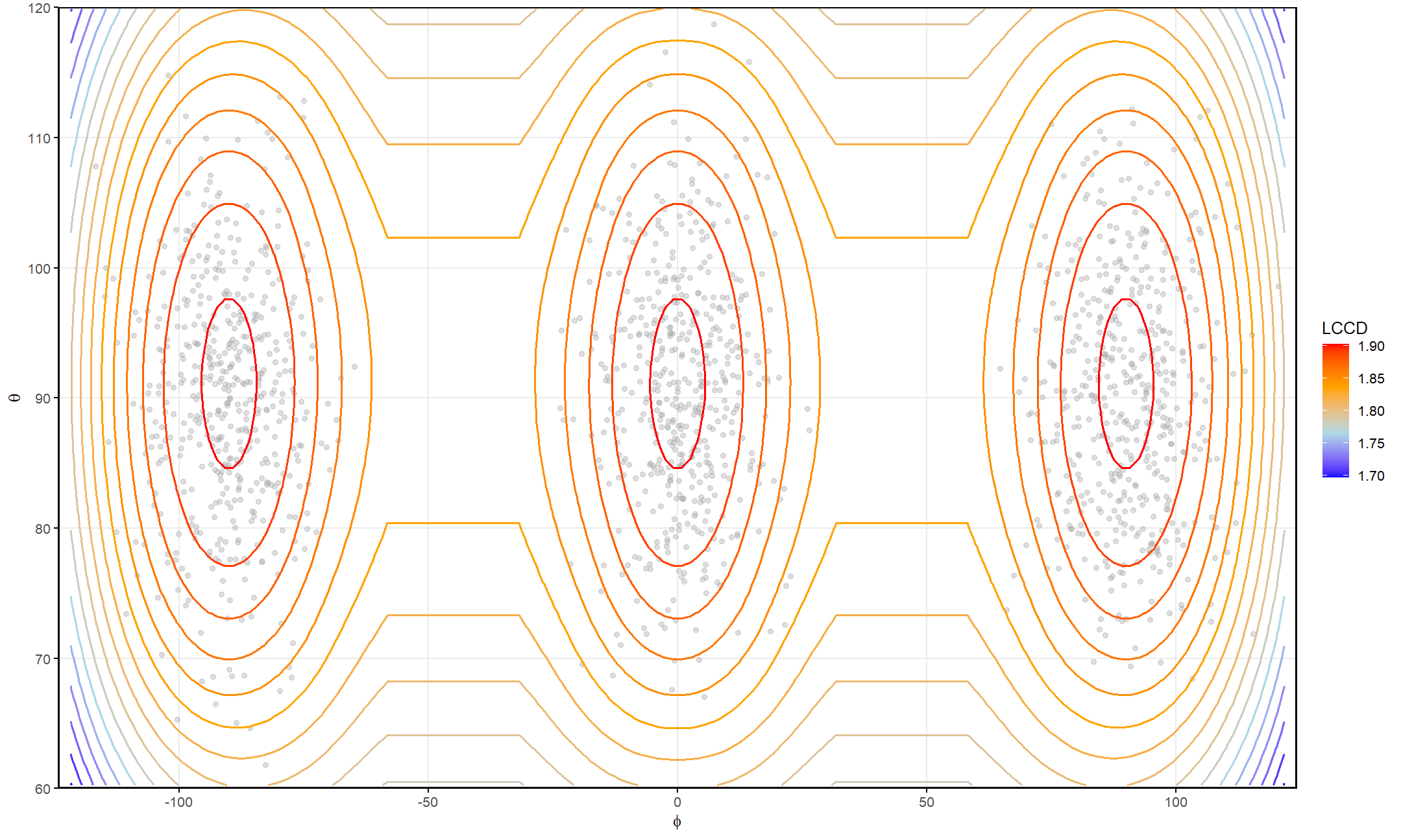}
\caption{}
\end{subfigure}

\caption{Density contour plot of a trimodal distribution on the sphere (a), along with the corresponding contour plots  CDD (b) and LCDD with $\beta=0.25$ (c).}
\label{fig:fig1}
\end{figure}

\begin{theorem}[Sample behaviour in $\beta$-neighborhoods]
\label{nesting}
For $0 < \beta_1 < \beta_2 \leq 1$,
\begin{enumerate}
\item $\mathrm{DN}^{(\beta_1)}_x \subseteq \mathrm{DN}^{(\beta_2)}_x$ \quad 
\item $\mathrm{LCDD}^{(\beta_1)}(x,X) \geq \mathrm{LCDD}^{(\beta_2)}(x,X)$ \quad 
\end{enumerate}
\end{theorem}

\begin{proof}
By Proposition~\ref{prop2}, we have $\mathrm{DN}^{(\beta)}_x = \mathrm{CN}^{(\beta)}_x$ for all $\beta$. 

Let $r_{(1)}(x) \leq r_{(2)}(x) \leq \cdots \leq r_{(n-1)}(x)$ denote the ordered cosine distances from $x$ to points in $X \setminus \{x\}$. 

Define $k_\beta = \lfloor \beta(n-1) \rfloor$, the integer part of $\beta(n-1)$. Since $\beta_1 < \beta_2$, we have $k_{\beta_1} < k_{\beta_2}$.\\

By definition,
\[
\mathrm{CN}^{(\beta_1)}_x = \{x_j \in X \setminus \{x\} : r_j(x) \leq r_{(k_{\beta_1})}(x)\}
\]
and
\[
\mathrm{CN}^{(\beta_2)}_x = \{x_j \in X \setminus \{x\} : r_j(x) \leq r_{(k_{\beta_2})}(x)\}.
\]
Since $r_{(k_{\beta_1})}(x) \leq r_{(k_{\beta_2})}(x)$, every point in $\mathrm{CN}^{(\beta_1)}_x$ also satisfies $r_j(x) \leq r_{(k_{\beta_2})}(x)$, hence belongs to $\mathrm{CN}^{(\beta_2)}_x$. Therefore,
\[
\mathrm{DN}^{(\beta_1)}_x = \mathrm{CN}^{(\beta_1)}_x \subseteq \mathrm{CN}^{(\beta_2)}_x = \mathrm{DN}^{(\beta_2)}_x.
\]

Recall the local cosine depth for a point $x$ at level $\beta$ is defined as
\[
\mathrm{LCDD}^{(\beta)}(x,X) = 2 - \frac{1}{k_\beta} \sum_{j=1}^{k_\beta} r_{(j)}(x).
\]
Let 
\[
S_{\beta_1} = \sum_{j=1}^{k_{\beta_1}} r_{(j)}(x) \quad \text{and} \quad S_{\beta_2} = \sum_{j=1}^{k_{\beta_2}} r_{(j)}(x).
\]

Since $k_{\beta_2} > k_{\beta_1}$, we can write $S_{\beta_2} = S_{\beta_1} + \Delta S$, where
\[
\Delta S = \sum_{j=k_{\beta_1}+1}^{k_{\beta_2}} r_{(j)}(x).
\]

Now consider the average distances:
\[
\frac{S_{\beta_1}}{k_{\beta_1}} = \frac{1}{k_{\beta_1}} \sum_{j=1}^{k_{\beta_1}} r_{(j)}(x)
\]
and
\[
\frac{S_{\beta_2}}{k_{\beta_2}} = \frac{1}{k_{\beta_2}} \left( S_{\beta_1} + \Delta S \right).
\]

Since $r_{(j)}(x)$ are non-decreasing, for any $j$ with $k_{\beta_1} < j \leq k_{\beta_2}$, we have $r_{(j)}(x) \geq r_{(k_{\beta_1})}(x)$. Moreover,
\[
r_{(k_{\beta_1})}(x) \geq \frac{S_{\beta_1}}{k_{\beta_1}},
\]
because the maximum of the first $k_{\beta_1}$ distances is at least their average. Therefore,
\[
\Delta S = \sum_{j=k_{\beta_1}+1}^{k_{\beta_2}} r_{(j)}(x) \geq (k_{\beta_2} - k_{\beta_1}) \cdot \frac{S_{\beta_1}}{k_{\beta_1}}.
\]

This inequality implies
\[
S_{\beta_1} + \Delta S \geq S_{\beta_1} + (k_{\beta_2} - k_{\beta_1}) \cdot \frac{S_{\beta_1}}{k_{\beta_1}}
= S_{\beta_1} \left( 1 + \frac{k_{\beta_2} - k_{\beta_1}}{k_{\beta_1}} \right)
= S_{\beta_1} \cdot \frac{k_{\beta_2}}{k_{\beta_1}}.
\]

Dividing by $k_{\beta_2}$ gives
\[
\frac{S_{\beta_2}}{k_{\beta_2}} \geq \frac{S_{\beta_1}}{k_{\beta_1}}.
\]

Finally, since $\mathrm{LCDD}^{(\beta)}(x,X) = 2 - \frac{S_\beta}{k_\beta}$, we obtain
\[
\mathrm{LCDD}^{(\beta_1)}(x,X) = 2 - \frac{S_{\beta_1}}{k_{\beta_1}} 
\geq 2 - \frac{S_{\beta_2}}{k_{\beta_2}} = \mathrm{LCDD}^{(\beta_2)}(x,X).
\]
\end{proof}

This theorem establishes two key properties: (i) neighborhoods expand as $\beta$ increases (nesting property), and (ii) local depth decreases as neighborhoods expand. This is intuitive: as we consider larger neighborhoods, the average distance to neighbors increases, reducing the measure of local centrality.

\subsection{Population version and properties}
\label{sec:population_lcdd}

While Definition~\ref{def:lcdd_sample} provides a computationally tractable measure for finite samples, 
theoretical analysis requires its population counterpart. 

\begin{definition}[Population Local Cosine Distance Depth]
\label{def:lcdd_population}
Let $F$ be a distribution on $S^{q-1}$ with continuous density $f$ 
that is bounded away from zero on its support. For $x \in S^{q-1}$ 
and $\beta \in (0,1]$, let $\rho^{(\beta)}(x)$ be the unique radius 
satisfying:
\[
F\left(\{y \in S^{q-1} : d_{\cos}(x,y) \leq \rho^{(\beta)}(x)\}\right) = \beta
\]

The \emph{population local cosine distance depth} of $x$ with respect 
to $F$ at locality level $\beta$ is defined as:
\[
\text{LCDD}^{(\beta)}(x, F) := 2 - \frac{1}{\beta} \bigg( E_{Y \sim F}
\left[d_{\cos}(x,Y) \cdot I\{d_{\cos}(x,Y) \leq \rho^{(\beta)}(x)\}\right]\bigg)
\]
\end{definition}

\begin{remark}
Equivalently, $\text{LCDD}^{(\beta)}(x, F) = \text{CDD}(x, F_{\beta,x})$, 
where $F_{\beta,x}$ denotes the conditional distribution of $F$ 
restricted to the geodesic ball of probability mass $\beta$ around $x$.
%radius $r^{(\beta)}(x)$ around $x$. 
This makes explicit the connection between local and global depth measures.
\end{remark}

Monotonicity in $\beta$ holds at the population level, as the conditional expectation of $d_{\cos}$ is nondecreasing in the radius $r^{(\beta)}(x)$.

\begin{proposition}[Continuity with respect to $\beta$]\label{prop:continuity}
Let $F$ be a distribution on $S^{q-1}$ with continuous density $f$ bounded away from zero on its support. For $\beta\in(0,1)$ and $\Delta\beta>0$ with $\beta+\Delta\beta\le1$,
\[
\sup_{x\in S^{q-1}} \big|\mathrm{LCDD}^{(\beta)}(x,F) - \mathrm{LCDD}^{(\beta+\Delta\beta)}(x,F)\big| = O(\Delta\beta).
\]
\end{proposition}

\begin{proof}
For fixed $x \in S^{q-1}$, let $\rho^{(\beta)}(x)$ be the unique radius satisfying
\[
F\left(\{y \in S^{q-1} : d_{\cos}(x,y) \leq \rho^{(\beta)}(x)\}\right) = \beta.
\]
Since $f$ is continuous and bounded away from zero, the quantile function $\beta \mapsto \rho^{(\beta)}(x)$ is Lipschitz continuous in $\beta$, uniformly in $x$.

Define the conditional expectations:
\[
\mu_\beta(x) = E[d_{\cos}(x,Y) \mid d_{\cos}(x,Y) \leq \rho^{(\beta)}(x)]
\]
and
\[
\mu_{\beta+\Delta\beta}(x) = E[d_{\cos}(x,Y) \mid d_{\cos}(x,Y) \leq \rho^{(\beta+\Delta\beta)}(x)].
\]

Let $A_x = \{y : \rho^{(\beta)}(x) < d_{\cos}(x,y) \leq \rho^{(\beta+\Delta\beta)}(x)\}$, so that $F(A_x) = \Delta\beta$. Then we can write:
\[
\mu_{\beta+\Delta\beta}(x) = \frac{\beta}{\beta+\Delta\beta} \mu_\beta(x) + \frac{\Delta\beta}{\beta+\Delta\beta} \mu_A(x),
\]
where $\mu_A(x) = E[d_{\cos}(x,Y) \mid Y \in A_x]$.

Rearranging gives:
\[
\mu_{\beta+\Delta\beta}(x) - \mu_\beta(x) = \frac{\Delta\beta}{\beta+\Delta\beta} (\mu_A(x) - \mu_\beta(x)).
\]

Now, since $\rho^{(\beta)}(x) \leq d_{\cos}(x,y) \leq \rho^{(\beta+\Delta\beta)}(x)$ for $y \in A_x$, we have:
\[
\mu_A(x) \in [\rho^{(\beta)}(x), \rho^{(\beta+\Delta\beta)}(x)].
\]
Also, by definition, $\mu_\beta(x) \leq \rho^{(\beta)}(x)$. Therefore,
\[
|\mu_A(x) - \mu_\beta(x)| \leq \rho^{(\beta+\Delta\beta)}(x) - \mu_\beta(x) \leq \rho^{(\beta+\Delta\beta)}(x) - 0 \leq 2,
\]
since $d_{\cos} \in [0,2]$. However, we can obtain a tighter bound using the continuity of $\rho^{(\beta)}(x)$.

From the Lipschitz continuity of $\rho^{(\beta)}(x)$ in $\beta$, there exists $L > 0$ such that for all $x$:
\[
|\rho^{(\beta+\Delta\beta)}(x) - \rho^{(\beta)}(x)| \leq L \Delta\beta.
\]
Thus,
\begin{align*}
|\mu_A(x) - \mu_\beta(x)| &\leq |\rho^{(\beta+\Delta\beta)}(x) - \mu_\beta(x)| 
\leq |\rho^{(\beta+\Delta\beta)}(x) - \rho^{(\beta)}(x)| + |\rho^{(\beta)}(x) - \mu_\beta(x)| \\
&\leq L\Delta\beta + |\rho^{(\beta)}(x) - \mu_\beta(x)|.
\end{align*}

The term $|\rho^{(\beta)}(x) - \mu_\beta(x)|$ is bounded by a constant $M$ uniformly in $x$ because $f$ is bounded away from zero, ensuring the conditional distribution is not too concentrated at the boundary. Hence,
\[
|\mu_A(x) - \mu_\beta(x)| \leq L\Delta\beta + M.
\]

Returning to the difference:
\[
|\mu_{\beta+\Delta\beta}(x) - \mu_\beta(x)| = \frac{\Delta\beta}{\beta+\Delta\beta} |\mu_A(x) - \mu_\beta(x)| 
\leq \frac{\Delta\beta}{\beta} (L\Delta\beta + M) = O(\Delta\beta).
\]

Since $\mathrm{LCDD}^{(\beta)}(x,F) = 2 - \mu_\beta(x)$, we have:
\[
|\mathrm{LCDD}^{(\beta)}(x,F) - \mathrm{LCDD}^{(\beta+\Delta\beta)}(x,F)| = |\mu_{\beta+\Delta\beta}(x) - \mu_\beta(x)| = O(\Delta\beta).
\]

The bound holds uniformly in $x$ because all constants ($L$, $M$, and the implicit constant in $O(\Delta\beta)$) are independent of $x$ due to the compactness of $S^{q-1}$ and the uniform bounds on $f$.
\end{proof}

\begin{corollary}[Limit behavior]\label{cor:limits}
Let $F$ be a distribution on $S^{q-1}$ with continuous density $f$ bounded away from zero on its support. Then for all $x\in S^{q-1}$,
\begin{align*}
\lim_{\beta\to0^+}~\mathrm{LCDD}^{(\beta)}(x,F) &= 2,\\
\lim_{\beta\to1^-}~\mathrm{LCDD}^{(\beta)}(x,F) &= \mathrm{CDD}(x,F),
\end{align*}
where $\mathrm{CDD}(x,F) = 2 - E_F[d_{\cos}(x,Y)]$ is the population cosine distance depth.
\end{corollary}

\begin{proof}
Recall that $\mathrm{LCDD}^{(\beta)}(x,F) = 2 - \mu_\beta(x)$, where
\[
\mu_\beta(x) = E[d_{\cos}(x,Y) \mid d_{\cos}(x,Y) \leq \rho^{(\beta)}(x)]
\]
and $\rho^{(\beta)}(x)$ satisfies $F\{d_{\cos}(x,Y) \leq \rho^{(\beta)}(x)\} = \beta$.

\textbf{(i)} As $\beta \to 0^+$, the radius $\rho^{(\beta)}(x) \to 0$ because the density $f$ is bounded away from zero. More precisely, since $f$ is continuous and positive, for small $\beta$ we have $\rho^{(\beta)}(x) = O(\beta)$.

For any $Y$ with $d_{\cos}(x,Y) \leq \rho^{(\beta)}(x)$, we have $0 \leq d_{\cos}(x,Y) \leq \rho^{(\beta)}(x) \to 0$. By dominated convergence (since $d_{\cos} \leq 2$), we obtain:
\[
\lim_{\beta\to0^+} \mu_\beta(x) = \lim_{\beta\to0^+} E[d_{\cos}(x,Y) \mid d_{\cos}(x,Y) \leq \rho^{(\beta)}(x)] = 0.
\]
Therefore,
\[
\lim_{\beta\to0^+} \mathrm{LCDD}^{(\beta)}(x,F) = 2 - 0 = 2.
\]

\textbf{(ii)} As $\beta \to 1^-$, the radius $\rho^{(\beta)}(x)$ increases toward 
\[
\rho_{\max}(x) = \inf\{t \geq 0 : F\{d_{\cos}(x,Y) \leq t\} = 1\},
\]
which is the smallest radius such that the ball of that radius around $x$ contains the entire support of $F$.

The conditional distribution given $d_{\cos}(x,Y) \leq \rho^{(\beta)}(x)$ converges weakly to the unconditional distribution $F$ as $\beta \to 1^-$. By the bounded convergence theorem (since $d_{\cos}$ is bounded),
\[
\lim_{\beta\to1^-} \mu_\beta(x) = E_F[d_{\cos}(x,Y)].
\]
Therefore,
\[
\lim_{\beta\to1^-} \mathrm{LCDD}^{(\beta)}(x,F) = 2 - E_F[d_{\cos}(x,Y)] = \mathrm{CDD}(x,F),
\]
where $\mathrm{CDD}(x,F)$ is the population version of the cosine distance depth defined as the expected value of $2 - d_{\cos}(x,Y)$.
\end{proof}

These limits have intuitive interpretations: as $\beta \to 0^+$, the neighborhood shrinks to a point, so the average distance goes to 0 and LCDD approaches its maximum value 2. As $\beta \to 1^-$, the neighborhood expands to cover the entire distribution, recovering the global CDD.

\begin{remark}
Under appropriate regularity conditions (continuity of $F$ and boundedness away from zero) the sample version (Definition~\ref{def:lcdd_sample}) is a uniformly consistent estimator of the population LCDD: as $n \to \infty$, the empirical $\beta(n-1)$ 
nearest neighbors converge to the population geodesic ball containing 
mass $\beta$, and the sample average of cosine distances converges to 
the population expectation.
\end{remark}

The following lemma establishes that this convergence is uniform over the hypersphere.

\begin{lemma}[Uniform consistency of LCDD]
\label{lem:uniform_consistency}
Let $F_n$ denote the empirical measure of a random sample $X_1, \ldots, X_n$ be i.i.d. from a distribution $F$ on $S^{q-1}$ with density $f$ that is continuous and bounded away from zero on its support. For any $\beta \in (0,1]$, we have
\[
\sup_{x \in S^{q-1}} \left| \mathrm{LCDD}^{(\beta)}(x, F_n) - \mathrm{LCDD}^{(\beta)}(x, F) \right| \xrightarrow{a.s.} 0
\quad\text{as } n \to \infty,
\]
where $\mathrm{LCDD}^{(\beta)}(x, F_n)$ denotes the sample LCDD based on $X_n = \{X_1,\dots,X_n\}$.
\end{lemma}
\begin{proof}
Let $k_n = \lfloor \beta(n-1) \rfloor$. For $x \in S^{q-1}$, let $r^{(1)}_n(x) \leq r^{(2)}_n(x) \leq \cdots \leq r^{(n)}_n(x)$ denote the ordered cosine distances $\{d_{\cos}(x,X_i)\}_{i=1}^n$.

Recall the definitions:
\[
\mathrm{LCDD}^{(\beta)}(x,F) = 2 - \mu_\beta(x), \quad \text{where} \quad 
\mu_\beta(x) = \frac{1}{\beta} E_F\big[d_{\cos}(x,Y) I\{d_{\cos}(x,Y) \leq \rho^{(\beta)}(x)\}\big],
\]
with $\rho^{(\beta)}(x)$ satisfying $F\{d_{\cos}(x,Y) \leq \rho^{(\beta)}(x)\} = \beta$.

The sample LCDD is:
\[
\mathrm{LCDD}^{(\beta)}(x,F_n) = 2 - \hat{\mu}_n(x), \quad \text{where} \quad
\hat{\mu}_n(x) = \frac{1}{k_n} \sum_{i=1}^{k_n} r^{(i)}_n(x).
\]

We prove uniform convergence in three steps.

\medskip
\noindent\textbf{Step 1} %: Uniform convergence of empirical quantiles.
Define the empirical process
\[
F_n(t;x) = \frac{1}{n} \sum_{i=1}^n I\{d_{\cos}(x,X_i) \leq t\},
\]
and its population counterpart $F(t;x) = P(d_{\cos}(x,Y) \leq t)$.

Since the class of sets $\{\{y: d_{\cos}(x,y) \leq t\}: x \in S^{q-1}, t \in [0,2]\}$ is a VC-class (as geodesic balls on the sphere), we have by the uniform Glivenko-Cantelli theorem:
\[
\sup_{x \in S^{q-1}} \sup_{t \in [0,2]} |F_n(t;x) - F(t;x)| \xrightarrow{a.s.} 0.
\]

Let $\rho^{(\beta)}(x)$ be the population $\beta$-quantile: $F(\rho^{(\beta)}(x);x) = \beta$.
Define the empirical quantile $\hat{\rho}_n(x) = r^{(k_n)}_n(x)$, which satisfies $F_n(\hat{\rho}_n(x);x) = k_n/n \to \beta$.

By the uniform continuity of $F(t;x)$ in $t$ (implied by $f$ being bounded away from zero) and uniform convergence of $F_n$, we obtain:
\[
\sup_{x \in S^{q-1}} |\hat{\rho}_n(x) - \rho^{(\beta)}(x)| \xrightarrow{a.s.} 0.
\]

\medskip
\noindent\textbf{Step 2} %: Uniform convergence of truncated empirical means.
Define the truncated empirical process:
\[
G_n(t;x) = \frac{1}{n} \sum_{i=1}^n d_{\cos}(x,X_i) I\{d_{\cos}(x,X_i) \leq t\},
\]
and its population counterpart $G(t;x) = E[d_{\cos}(x,Y) I\{d_{\cos}(x,Y) \leq t\}]$.

The class of functions
\[
\mathcal{F} = \{(x,y) \mapsto d_{\cos}(x,y) I\{d_{\cos}(x,y) \leq t\}: x \in S^{q-1}, t \in [0,2]\}
\]
is uniformly bounded (by 2) and, as a product of a Lipschitz function $d_{\cos}(x,\cdot)$ with an indicator of a VC-class, is itself a Glivenko-Cantelli class. Therefore,
\[
\sup_{x \in S^{q-1}} \sup_{t \in [0,2]} |G_n(t;x) - G(t;x)| \xrightarrow{a.s.} 0.
\]

Now, by the continuous mapping theorem and Step 1,
\[
G_n(\hat{\rho}_n(x);x) \xrightarrow{a.s.} G(\rho^{(\beta)}(x);x) = \beta \mu_\beta(x)
\]
uniformly in $x$. But also,
\[
G_n(\hat{\rho}_n(x);x) = \frac{1}{n} \sum_{i=1}^n d_{\cos}(x,X_i) I\{d_{\cos}(x,X_i) \leq \hat{\rho}_n(x)\} = \frac{k_n}{n} \cdot \hat{\mu}_n(x).
\]
Since $k_n/n \to \beta$ almost surely, we conclude:
\[
\sup_{x \in S^{q-1}} |\hat{\mu}_n(x) - \mu_\beta(x)| \xrightarrow{a.s.} 0.
\]

\medskip
\noindent\textbf{Step 3} %: Completion of the proof.}
From Steps 1 and 2, we have:
\[
\sup_{x \in S^{q-1}} |\mathrm{LCDD}^{(\beta)}(x,F_n) - \mathrm{LCDD}^{(\beta)}(x,F)| 
= \sup_{x \in S^{q-1}} |\hat{\mu}_n(x) - \mu_\beta(x)| \xrightarrow{a.s.} 0.
\]
This establishes the desired uniform almost sure convergence.
\end{proof}

\begin{corollary}
\label{cor:condition_a2}
Under the assumptions of Lemma~\ref{lem:uniform_consistency}, for any 
$\beta^* \in (0,1]$:
\begin{equation*}
\sup_{x \in S^{q-1}} \left| LCDD^{(\beta^*)}(x,\hat{F}_n) - LCDD^{(\beta^*)}(x,F) 
\right| \xrightarrow{p} 0
\end{equation*}
\end{corollary}

Note that the continuity result in Proposition \ref{prop:continuity} does not hold in the empirical case, where the LCDD is in general a piecewise constant function of $\beta$. Specifically, for a fixed point $x$ and $k \in \{1,2,\ldots,n-1 \}$, the $\mathrm{LCDD}^{(\beta)}(x,\hat F_n)$ remains constant on each interval $\big[\tfrac{k}{n-1}, \tfrac{k+1}{n-1}\big)$, while continuity at $\beta = k/(n-1)$ holds if and only if the cosine distance of the $k$ nearest points coincides with the average of the previous distances. Since this equality is highly unlikely in non-degenerate samples, the empirical LCDD is typically discontinuous in $\beta$.

\section{DD-classifier with local depth}
\label{sec:localDD}

The depth of a given point characterizes its location w.r.t. the whole distribution. Thus the classifiers which use any global depth function perform well only if the considered distributions
have some global properties like symmetry or unimodality. To obtain good performance also in more general settings, the use of some local depth should be preferred. The problem that emerged and need to be handle is choice of localization level.
The first classifier which employed local depth was proposed in 2013 by \cite{Hlubinka2013depth} who used a weighted halfspace depth. Later on, \cite{paindaveine2013depth} developed the more
complex approach we take inspiration from for our proposal. However, the just cited works which enables localization of any global depth function which is then used in the maximum depth classifier. 

Here, we use the proposed local cosine distance depth to be applied in the DD-plot, where then a polynomial separating function is adopted to discriminate between groups.

We focus on two-class classification problem. Let $\{X_{1},\ldots, X_{m}\}(\equiv X)$ and $\{Y_{1},\ldots, Y_{n}\}(\equiv  Y)$ be two random samples from $F_{1}$ and $F_{2}$, respectively, which are distributions defined on $S^{q-1}$ . As seen in \cite{pandolfo2022gld} and from the definition of the DD-plot, if $F_1 = F_2$, then DD-plot should be concentrated along the 45-degree line. Conversely, if the two distributions differ, the DD-plot would exhibit a departure from the 45-degree line.

Hence, given a locality level $\beta$, the proposed classifier is then defined as
\[
C_{\beta,s}(x)
=
\begin{cases}
2, & \text{if } LCDD^{(\beta)}(x,\hat F_{2})
\geq s\bigl(LCDD^{(\beta)}(x,\hat F_{1})\bigr),\\[1ex]
1, & \text{otherwise.}
\end{cases}
\]

For any given $s(\cdot)$, we then draw a curve corresponding to $y = s(x)$ in the DD-plot, and assign the observations above the curve to $F_{1}$ and those below it to $F_{2}$, and then calculate the empirical misclassification rate, that is
\begin{equation}
\begin{split}
\hat{R}_s &= \frac{\pi_1}{m} \sum_{i=1}^m I_{\{LCDD^{(\beta)}(X_i,\hat{F}_1) \leq s(LCDD^{(\beta)}(X_i,\hat{F}_2))\}} \\
&+\frac{\pi_2}{n} \sum_{i=1}^n I_{\{LCDD^{(\beta)}(Y_i,\hat{F}_1) > s(LCDD^{(\beta)}(Y_i,\hat{F}_2))\}}. \label{eq:misclrate}
\end{split}
\end{equation}
Where $\pi_{i}$ are the prior probabilities of the two classes, $N = (m, n)$, and $I_{A}$ is the indicator function which takes $1$ if $A$ is true and $0$ otherwise. Hence, $s(\cdot)$ is estimated to minimize $\hat{R}_s$. Following \cite{Li01062012}, here we consider  
$$
s(x) = \sum_{i=1}^{k_{0}}a_{i}x^{i},
$$ 
where $k_{0}$ is the given degree of the polynomial and $\mathbf{a} = (a_{1},\ldots,a_{k_{0}}) \in \mathbb{R}^{k_{0}}$ is the coefficient vector of the polynomial.

\begin{theorem}[Bayes Consistency of LCDD-DD Classifier]
\label{thm:bayes_consistency}
Let $F_1, F_2$ be distributions on $S^{q-1}$ with continuous densities $f_1, f_2$ bounded away from zero, and let $\pi_1, \pi_2 > 0$ be class priors with $\pi_1 + \pi_2 = 1$. 
Let $\mathcal{R}$ be a compact class of functions (pointwise compact) consisting of polynomials of degree at most $k_0$.

Define the empirical LCDD depths $D^{(\beta)}_{\hat{F}_j}(z) := \mathrm{LCDD}^{(\beta)}(z, \hat{F}_j)$ for $j=1,2$, where $\hat{F}_1, \hat{F}_2$ are the empirical distributions from samples of sizes $n_1, n_2$ respectively, with total sample size $N = n_1 + n_2$. Consider the classifier:
\[
C_{\beta,s}(z) = 
\begin{cases}
2 & \text{if } D^{(\beta)}_{\hat{F}_2}(z) \geq s\big(D^{(\beta)}_{\hat{F}_1}(z)\big) \\
1 & \text{otherwise}
\end{cases}
\]

Assume the following conditions hold:

\begin{itemize}\setlength{\itemindent}{0.5cm}
\item[(A1)] There exists a unique $\beta^* \in (0,1]$ and $s_B \in \mathcal{R}$ such that the classifier $C_{\beta^*,s_B}$ equals the Bayes classifier almost everywhere, i.e.,
\[
C_{\beta^*,s_B}(z) = I\{\pi_2 f_2(z) > \pi_1 f_1(z)\} \quad \text{a.e.}
\]

\item[(A2)] For $j=1,2$ and any $\beta \in (0,1]$,
\[
\sup_{z \in S^{q-1}} \big|D^{(\beta)}_{\hat{F}_j}(z) - D^{(\beta)}_{F_j}(z)\big| \xrightarrow{p} 0
\]
as $n_j \to \infty$.

\item[(A3)] $\mathcal{R}$ is compact in the topology of pointwise convergence.

\item[(A4)] For each $N$, let $B_N \subset (0,1]$ be a finite grid such that 
\[
\max_{\beta, \beta' \in B_N} |\beta - \beta'| \to 0 \quad \text{as } N \to \infty.
\]
Define $\hat{s}_{N,\beta} = \underset{s \in \mathcal{R}}{\argmin}~\widetilde{R}_N(s,\beta)$, where $\widetilde{R}_N$ is the empirical risk, and select
\[
\hat{\beta}_N = \argmin_{\beta \in B_N} \mathrm{CV}(\hat{s}_{N,\beta}, \beta),
\]
where $\mathrm{CV}$ denotes cross-validated error. Finally, set $\hat{s}_N = \hat{s}_{N,\hat{\beta}_N}$ and $\hat{C}_N = C_{\hat{\beta}_N, \hat{s}_N}$.
\end{itemize}

Then, as $N \to \infty$ with $n_j/N \to \lambda_j \in (0,1)$ for $j=1,2$, we have:

\begin{itemize}\setlength{\itemindent}{0.5cm}
\item[(i)] $\hat{\beta}_N \xrightarrow{p} \beta^*$;
\item[(ii)] $\hat{s}_N \xrightarrow{p} s_B$ pointwise;
\item[(iii)] $R(\hat{C}_N) \xrightarrow{p} R_{\mathrm{Bayes}} = R(C_{\beta^*,s_B})$.
\end{itemize}
Assumption (A1) requires that there exists some locality level $\beta^*$ for which the LCDD-based classifier achieves the Bayes error. This is a reasonable assumption when the data have local structure that can be captured at an appropriate scale. In practice, even if the exact Bayes classifier is not achievable with polynomial separators in the LCDD space, the theorem guarantees that our data-driven procedure will approach the best possible performance within the class $\mathcal{R}$.
\end{theorem}

\begin{proof}
We adapt the framework of \cite{Li01062012} to our LCDD-based classifier with data-driven selection of the locality parameter $\beta$.

\medskip
\noindent\textbf{Step 1:} %: Continuity in the locality parameter.}
By Proposition \ref{prop:continuity}, for any $\beta \in (0,1)$ and $\Delta\beta > 0$ with $\beta + \Delta\beta \leq 1$,
\[
\sup_{x \in S^{q-1}} \big|\mathrm{LCDD}^{(\beta)}(x,F_j) - \mathrm{LCDD}^{(\beta+\Delta\beta)}(x,F_j)\big| = O(\Delta\beta), \quad j=1,2,
\]
uniformly in $x$. In particular, for each fixed $z$, the map $\beta \mapsto D^{(\beta)}_{F_j}(z)$ is continuous on $(0,1]$.

The misclassification error of the classifier $C_{\beta,r}$ is:
\[
R(\beta,s) = \pi_1 P_{F_1}(C_{\beta,s}(Z) = 2) + \pi_2 P_{F_2}(C_{\beta,s}(Z) = 1).
\]
Since the classifier depends on $z$ only through $(D^{(\beta)}_{F_1}(z), D^{(\beta)}_{F_2}(z))$ and the indicator of the decision region, the dominated convergence theorem implies that for each fixed $s \in \mathcal{R}$, the function $\beta \mapsto R(\beta,s)$ is continuous on $(0,1]$.

\medskip
\noindent\textbf{Step 2:} 
%: Uniform convergence of empirical processes.} 
Fix $\beta \in (0,1]$ and $j \in \{1,2\}$. By Lemma \ref{lem:uniform_consistency} (which implies assumption A2),
\[
\sup_{z \in S^{q-1}} \big|D^{(\beta)}_{\hat{F}_j}(z) - D^{(\beta)}_{F_j}(z)\big| \xrightarrow{a.s.} 0.
\]
Since $B_N$ is finite for each $N$, a union bound yields:
\[
\max_{\beta \in B_N} \sup_{z \in S^{q-1}} \big|D^{(\beta)}_{\hat{F}_j}(z) - D^{(\beta)}_{F_j}(z)\big| \xrightarrow{a.s.} 0 \quad \text{as } n_j \to \infty.
\]

Consider the class of decision sets in the depth space:
\[
\mathcal{D} = \big\{\{(u,v) \in [0,2]^2 : v \leq s(u)\} : s \in \mathcal{R}\big\}.
\]
Since $\mathcal{R}$ consists of polynomials of bounded degree, $\mathcal{D}$ has finite VC dimension \citep[Theorem 4]{Li01062012}. Therefore, $\mathcal{D}$ is a Glivenko-Cantelli class, which implies:
\[
\sup_{\beta \in B_N} \sup_{s \in \mathcal{R}} \big|\widehat{R}_N(s,\beta) - R(\beta,s)\big| \xrightarrow{p} 0,
\]
where $\widehat{R}_N(s,\beta)$ is the empirical risk of $C_{\beta,s}$.

\medskip
\noindent\textbf{Step 3:} 
% Consistency of $\hat{\beta}_N$.}
Define the optimal risk function:
\[
R^*(\beta) = \inf_{s \in \mathcal{R}} R(\beta,s), \quad \beta \in (0,1].
\]
From Step 1, $R(\beta,s)$ is continuous in $\beta$ for each $s$, and by compactness of $\mathcal{R}$, the infimum is attained and $R^*(\beta)$ is continuous. Assumption (A1) implies that $\beta^*$ is the unique minimizer of $R^*(\beta)$ over $(0,1]$, with $R^*(\beta^*) = R_{\mathrm{Bayes}}$.

For each $\beta \in B_N$, $\hat{s}_{N,\beta}$ minimizes the empirical risk. By the uniform convergence in Step 2 and standard M-estimation theory:
\[
\sup_{\beta \in B_N} \big|R(\beta,\hat{s}_{N,\beta}) - R^*(\beta)\big| \xrightarrow{p} 0.
\]

Since $R^*$ is continuous with unique minimizer $\beta^*$, and $B_N$ becomes dense in $(0,1]$ as $N \to \infty$ by (A4), we obtain:
\[
\hat{\beta}_N = \underset{\beta \in B_N}{\argmin}~ \mathrm{CV}(\hat{s}_{N,\beta},\beta) \xrightarrow{p} \beta^*.
\]

\medskip
\noindent\textbf{Step 4:} 
% Consistency of $\hat{s}_N$ and risk convergence.}
For fixed $\beta = \beta^*$, the empirical risk minimizer satisfies $\hat{s}_{N,\beta^*} \xrightarrow{p} s_B$ by standard consistency results for classification with VC classes. The convergence $\hat{\beta}_N \xrightarrow{p} \beta^*$ and continuity of the risk function imply:
\[
\hat{s}_N = \hat{s}_{N,\hat{\beta}_N} \xrightarrow{p} s_B.
\]

Finally, for the risk convergence:
\[
|R(\hat{C}_N) - R_{\mathrm{Bayes}}| \leq |R(\hat{\beta}_N,\hat{s}_N) - R(\hat{\beta}_N,s_B)| + |R(\hat{\beta}_N,s_B) - R(\beta^*,s_B)|.
\]
The first term converges to 0 because $\hat{s}_N \xrightarrow{p} s_B$ and $r \mapsto R(\beta,s)$ is continuous uniformly in $\beta$ on compacts. The second term converges to 0 because $\hat{\beta}_N \xrightarrow{p} \beta^*$ and $\beta \mapsto R(\beta,s_B)$ is continuous. Hence,
\[
R(\hat{C}_N) \xrightarrow{p} R(\beta^*,s_B) = R_{\mathrm{Bayes}}.
\]
\end{proof}

In the following sections, we investigate the practical performance of the LCDD-DD classifier through simulations (Section~\ref{simulation}) and real-data applications (Section~\ref{applications}). We compare it with global depth-based classifiers and other directional classification methods

\section{Simulations}
\label{simulation}

Among the various applications of data depth, supervised classification represents the most prominent one, particularly in the context of directional data, where the absence of a natural ordering poses specific challenges.
In this study, the proposed local depth function is evaluated against its global counterpart through a simulation experiment, in which both serve as the underlying measures for DD-classifier training. 
This section presents the simulation study designed to compare the performance of the Global and Local CDD when incorporated into the DD-classifier. 
Two distinct simulation scenarios are considered, each described in detail in the following subsections and including three experimental setups.

\subsection{The simulation design}
The first simulation scenario aims to assess the potential classification improvement achieved by the local depth compared to its global counterpart when observations belonging to the same class are distributed across multiple clusters. 
The second simulation scenario investigates a different type of data distribution, characterized by a non-convex structure combined with a strongly pronounced multimodality.

For each setup, we generated $100$ datasets of size $n = 500$. 
The neighborhood parameter, expressed as the proportion  $\beta$ of nearest units within each class, takes values $\beta \in \{0.05, 0.10, 0.25\}$, 
while the data dimensionality varies across $d \in \{3, 10, 25\}$.
In accordance with \cite{Guyon1997ASL}, 70\% of the observations were allocated to the training set and the remaining 30\% to the test set.

In this study we focus on the binary classification setting with the implementation of DD-classifier. 
Let $W_{1i}$, $i = 1, \dots, n_1$, and $W_{2i}$, $i = 1, \dots, n_2$, denote independent random samples drawn from the distributions $F_1$ and $F_2$ on $S^{q-1}$, respectively, where $n_1$ and $n_2$ are the cardinality of the two classes. 
The proportion between the two classes was randomly chosen to lie between 35\% and 50\%. 

The evaluation metric considered for each simulated dataset is the misclassification rate (MR), as defined in eq. \ref{eq:misclrate}.

The simulation was fully implemented in \textsf{R}, using the \textit{ddalpha} package \citep{JSSv091i05}, which allows the customization of the DD-classifier training procedure to incorporate the proposed depth function and to automatically select the polynomial degree $p$ through a cross-validation scheme performed in the depth space.

\subsubsection{Scenario 1}\label{vMF}

The data of the first simulated scenario were generated according to the von Mises--Fisher (vMF) distribution, which plays for data on the unit hypersphere $S^{q-1}$ the same role as the normal distribution does for unconstrained Euclidean data. 
A $(q-1)$-dimensional unit random vector $x \in S^{q-1}$ is said to follow a vMF distribution if its probability density function is $f_q(x \mid \mu, \kappa) = C_q(\kappa)\exp(\kappa \mu' x)$,
where $\|\mu\| = 1$, $\kappa \geq 0$, and $q \geq 2$. 
The normalizing constant is given by
$$C_q(\kappa) = \frac{\kappa^{q/2 - 1}}{(2\pi)^{q/2} I_{q/2 - 1}(\kappa)},$$
where $I_b$ denotes the modified Bessel function of the first kind and order $b$.
The vMF distribution is characterized by the mean direction $\mu$ and the concentration parameter $\kappa$, which controls the degree of dispersion of the observations around the mean vector. 
In the limiting cases, when $\kappa = 0$, the distribution reduces to the uniform density on $S^{q-1}$, whereas as $\kappa \to \infty$ it degenerates into a point mass at $\mu$.
    
For the simulation study, three experimental setups were considered. 
For each setup, the concentration parameter $\kappa$ was randomly drawn to induce low, medium, and high noise levels in the data. 
Specifically, $\kappa_{\text{low}} \sim U[15,17]$, $\kappa_{\text{medium}} \sim U[10,12]$, and $\kappa_{\text{high}} \sim U[5,7]$. 

Each class was generated as a mixture of vMF distributions, each one indicated as $F_{(\mu_j^w,\kappa)}^{vMF}$. 
The choice of each center $\mu_j^w$, where $j$ denotes the component within class $w$, was made randomly but constrained to lie at specific cosine distances $d_{cos}(\cdot,\cdot)$ from the other centers. 
Starting from the same initial center $\mu_1^1 = (\eta_1, \dots, \eta_{q-1}, \eta_q)$, where $\eta_1 = 1$ and $\eta_t = 0$ for all $t = 2, \dots, q$, 
the remaining centers were generated according to the following setups:

\begin{itemize}
    \item \textbf{Setup 1:} 
    The second class center $\mu_1^2$ is randomly generated and constrained to satisfy $d_{cos}(\mu_1^1, \mu_1^2) \in [0.3, 0.5]$. 
    Thus, points for each class are drawn from vMF distributions with different mean directions but equal concentration parameter $\kappa$: 
    $F_1 = F_{(\mu_1^1, \kappa)}^{\text{vMF}}$ and $F_2 = F_{(\mu_1^2, \kappa)}^{\text{vMF}}$.

    \item \textbf{Setup 2:} 
    The second component of the first class, $\mu_2^1$, is randomly generated under the constraint $d_{cos}(\mu_1^1, \mu_2^1) \in [0.6, 0.8]$. 
    Then, $\mu_1^2$ is generated so that $d_{cos}(\mu_1^1, \mu_1^2) = d_{cos}(\mu_2^1, \mu_1^2) \in [0.25, 0.45]$. 
    Finally, $\mu_2^2$ is generated under the constraints 
    $d_{cos}(\mu_2^2, \mu_1^2) \in [d_{cos}(\mu_1^1, \mu_2^1) - \epsilon, d_{cos}(\mu_1^1, \mu_2^1) + \epsilon]$ and 
    $d_{cos}(\mu_2^1, \mu_2^2) \in [d_{cos}(\mu_1^1, \mu_1^2) - \epsilon, d_{cos}(\mu_1^1, \mu_1^2) + \epsilon]$, 
    where $\epsilon = 0.1$. 
    In this case, each class is modeled as an equally weighted mixture of two vMF distributions with different mean directions and the same concentration parameter $\kappa$: 
    $F_1 = \tfrac{1}{2} F_{(\mu_1^1, \kappa)}^{\text{vMF}} + \tfrac{1}{2} F_{(\mu_2^1, \kappa)}^{\text{vMF}}$ and 
    $F_2 = \tfrac{1}{2} F_{(\mu_1^2, \kappa)}^{\text{vMF}} + \tfrac{1}{2} F_{(\mu_2^2, \kappa)}^{\text{vMF}}$.

    \item \textbf{Setup 3:} 
    The second component of the first class, $\mu_2^1$, is randomly generated such that $d_{cos}(\mu_1^1, \mu_2^1) \in [0.4, 0.6]$. 
    Then, $\mu_1^2$ is generated under the constraints $d_{cos}(\mu_1^1, \mu_1^2) \in [0.4, 0.6]$ and $d_{cos}(\mu_2^1, \mu_1^2) \in [0.8, 1]$. 
    Finally, $\mu_2^2$ is generated so that $d_{cos}(\mu_1^1, \mu_2^2) \in [0.4, 2]$, 
    $d_{cos}(\mu_2^1, \mu_2^2) \in [0.4, 2]$, and $d_{cos}(\mu_1^2, \mu_2^2) \in [0.8, 2]$. 
    Similarly to Setup 2, each class is represented by an equally weighted mixture of two vMF components: 
    $F_1 = \tfrac{1}{2} F_{(\mu_1^1, \kappa)}^{\text{vMF}} + \tfrac{1}{2} F_{(\mu_2^1, \kappa)}^{\text{vMF}}$ and 
    $F_2 = \tfrac{1}{2} F_{(\mu_1^2, \kappa)}^{\text{vMF}} + \tfrac{1}{2} F_{(\mu_2^2, \kappa)}^{\text{vMF}}$.
\end{itemize}

The outcomes of this scenario are displayed in Figure \ref{fig:sc1}, which reports the distributions of the misclassification rates across setups, dimensions, and noise levels. 
A detailed interpretation of these findings is provided in \ref{res}.

\subsubsection{Scenario 2}
The Watson distribution was selected as the generating model for the second scenario of the simulation study. 
Although originally defined for axial data, the Watson distribution can effectively produce non-convex and bipolar structures, making it suitable for assessing the behavior of local depth functions in complex directional settings.

A random unit vector $x \in S^{q-1}$ follows a Watson distribution if its probability density function is given by 
\[
f(x \mid \mu, \kappa) = M\left(\frac{1}{2},\frac{q}{2}, \kappa \right)^{-1} \exp\{\kappa (\mu' x)^2 \},
\]
where $M(1/2, q/2, \cdot)$ denotes the Kummer's function.

The parameter $\mu$ represents the mean axis, while $\kappa$ determines the axial concentration of the data.
For $\kappa > 0$, the distribution is bipolar, and as $\kappa$ increases, it becomes increasingly concentrated around $\pm \mu$.
For $\kappa < 0$, it becomes a symmetric girdle distribution, with data concentrated in the subspace orthogonal to $\mu$, and the degree of dispersion is governed by the magnitude of $\kappa$.

As in the first scenario three setups were built using the Watson distribution, indicated with $F_{(\mu, \kappa)}^{Wat}$, and random values of $\kappa$ were used to produce low, medium, and high levels of noise in the data. 
Specifically, $\kappa_{\text{low}} \sim U[15,17]$, $\kappa_{\text{medium}} \sim U[10,12]$, and $\kappa_{\text{high}} \sim U[5,7]$.

In each setup, the mean axes of the two classes are randomly generated and constrained to have a cosine distance between $0.5$ and $0.7$, starting from an initial center defined as $\mu^1 = (\eta_1, \dots, \eta_{q-1}, \eta_q)$, where $\eta_1 = 1$ and $\eta_t = 0$ for all $t = 2, \dots, q$. 
To evaluate different data configurations, the sign of the concentration parameter $\kappa$ varies across the setups:

\begin{itemize}
    \item \textbf{Setup 1:} The parameter $\kappa$ is positive for both classes, resulting in two bipolar distributions: $F_1= F_{(\mu^1, \kappa)}^{Wat}$ and $F_2= F_{(\mu^2, \kappa)}^{Wat}$.
    \item \textbf{Setup 2:} The parameter $\kappa$ is negative for both classes, producing two girdle-shaped distributions: $F_1= F_{(\mu^1, -\kappa)}^{Wat}$ and $F_2= F_{(\mu^2, -\kappa)}^{Wat}$.
    \item \textbf{Setup 3:} The parameter $\kappa$ is positive for one class and negative for the other, leading to two response populations with different shapes: $F_1= F_{(\mu^1, \kappa)}^{Wat}$ and $F_2= F_{(\mu^2, -\kappa)}^{Wat}$.
\end{itemize}

The results of this scenario are reported in Figure \ref{fig:sc2}. 
The distributions of the misclassification rates are displayed, further conditioned by the setup, data dimension, and noise level. 
A detailed discussion of these results is provided in \ref{res}.

\begin{figure}[h!]
    \centering
    \includegraphics[width=\linewidth]{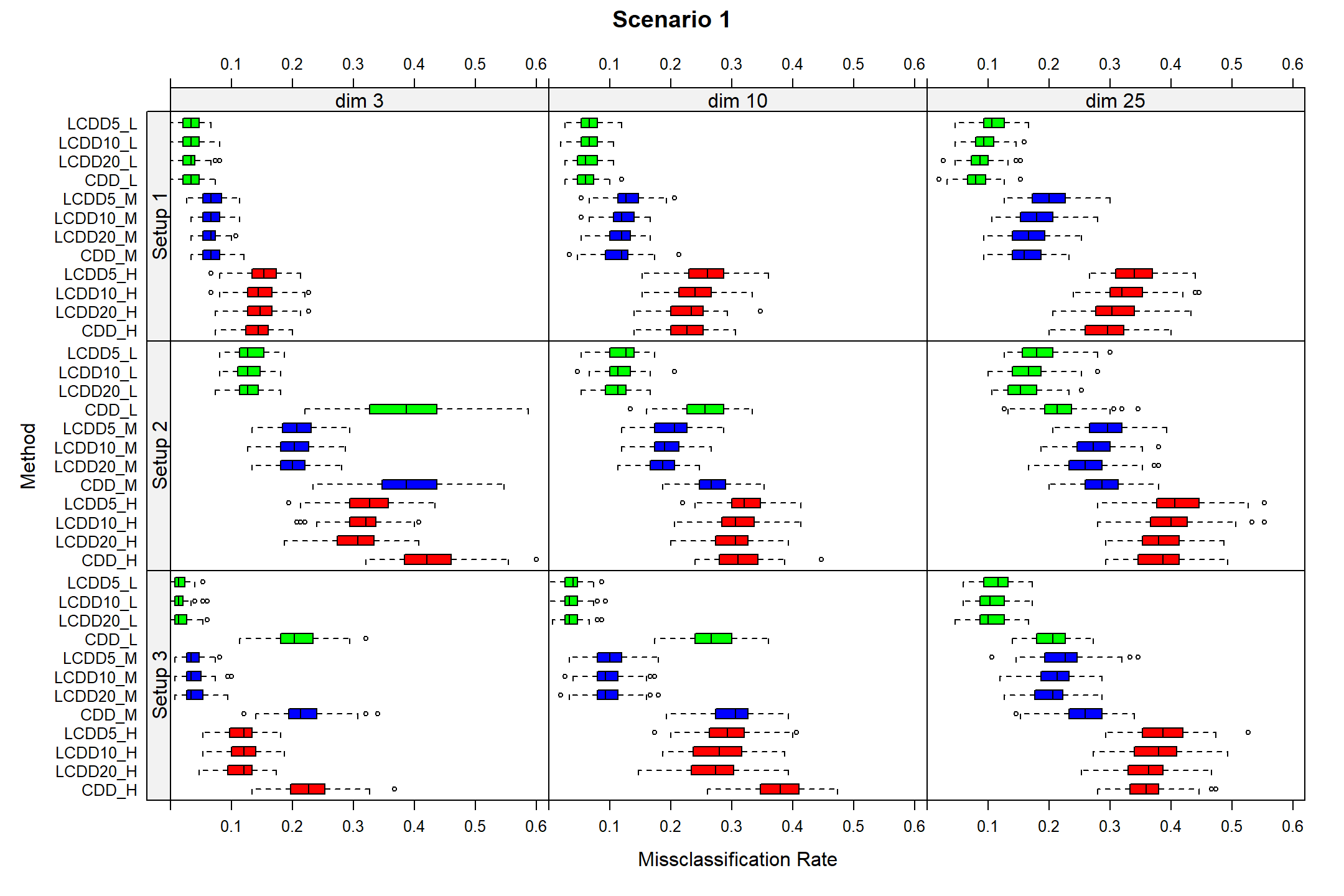}
    \caption{Simulation results for Scenario 1. The rows of the table indicate the specific Setup, the columns the number of dimensions, and the three different colors indicate noise levels: Low (L)--green, Medium (M)--blue, High (H)--red. }
    \label{fig:sc1}
\end{figure}

\subsection{Results}
\label{res}

Considering the results shown in Figure \ref{fig:sc1} and Figure \ref{fig:sc2} for the first and second scenario, respectively, the prediction error consistently increases with both the data dimension and the noise level, regardless of the specific neighborhood size (including the CDD). 
In both scenarios, there is effectively no difference between the local and global approaches when the dimension is $25$ and the noise level is high, as their classification errors become almost indistinguishable. 
Moreover, within the local framework, the three neighborhood proportions considered ($5\%$, $10\%$, and $25\%$) yield very similar performance.

Focusing on the first scenario, which relies on the vMF distribution, in the first setup, that does not involve multimodality or non-convexity, all methods exhibit nearly identical performances. 
Nevertheless, when the dispersion becomes very high, the two classes tend to overlap, making it slightly more advantageous to use larger neighborhood proportions for the depth computation, although the improvement is modest.

\begin{figure}[h!]
    \centering
    \includegraphics[width=\linewidth]{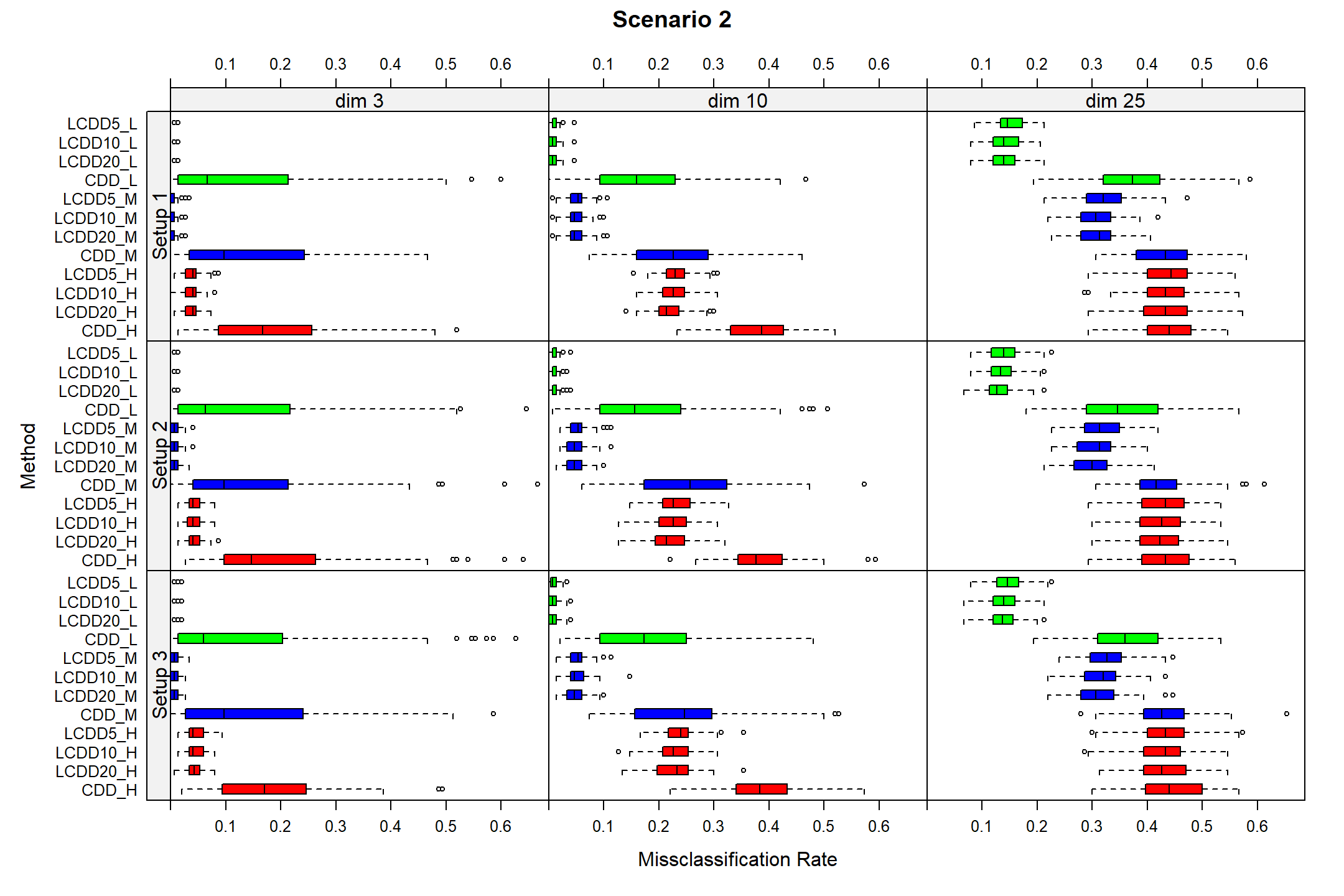}
    \caption{Simulation results for Scenario 2. The rows of the table indicate the specific Setup, the columns the number of dimensions, and the three different colors indicate noise levels: Low (L)--green, Medium (M)--blue, High (H)--red.}
    \label{fig:sc2}
\end{figure}

As the structure becomes more complex, as in the second setup, the CDD begins to display a higher prediction error compared to the LCDD, indicating that a local perspective is preferable in this setting, particularly under low noise. 
For higher noise levels and $10$ dimensions or more, however, the performance of all methods becomes essentially indistinguishable.

In the third setup of the same scenario, except for the previously discussed cases regarding the general behaviour in high dimensions and high noise, the global approach consistently underperforms the LCDD.
In particular, for $3$ and $10$ dimensions under low noise, the local depth achieves a misclassification error below $1\%$, demonstrating excellent performance even in the presence of a highly structured and challenging data configuration.

The classification performance of the CDD was previously investigated by \cite{pandolfo2021depth}, who compared it with other depth measures and showed that it performs very well in several settings involving the vMF distribution. 
However, in the context of the second scenario, which is based on the Watson distribution instead, the results are straightforward to interpret: the CDD consistently underperforms compared to the LCDD under all imposed conditions. 
This indicates that the global CDD lacks the flexibility required to adapt to scenarios in which classes are not well separated, as is the case in this simulation design.

To further enrich the comparison between the CDD and the LCDD, the next section presents an application to real datasets. 
This additional analysis allows us to assess the practical performance of the local depth approach in real-world situations and highlights its potential as an effective and flexible strategy for directional-data classification through depth-based methods.

\section{Real data examples} 
\label{applications}

In this section we compare the performance of the proposed classifier in comparison with its global version by means of two real-world datasets.  
We ran a 10-fold cross-validation repeated on 10 different training sets \citep{paindaveine2013depth} in order to select the best possible value of $\beta$ in the set $\{0.01,0.05,0.1,0.25,0.5,1\}$ according to the lowest average misclassification rate (MR). 

\subsection{Wholesale customers}
\label{whole}
The first real data refers to clients of a wholesale distributor. When dealing with marketing applications, the target groups can themselves be composed of different segments of the population, which usually present different spending habits. Thus, it could be more efficient to introduce more flexibility when classifying these units through the use of depth
 functions, allowing for a more local focus.  There are a total of 440 observations and 7 variables. The first two are categorical variables. Of these two, we are interested in the variable Channel, which can be either  Horeca (Hotel/Restaurant/Cafè) or Retail channel (Nominal), and will define our two classes for this problem.
The last 5 variables have information about the annual spending in monetary units on diverse product categories: (1) fresh products, (2) milk products
and (3), (4) frozen products, (5) detergents and paper products and (6)
delicatessen products. 
We treat this data as compositional data, exploiting the square-root transformation so that the points lie on a unit hypersphere. As apparent from Fig. \ref{fig:whole}, here the repeated cross-validation will select $\beta=0.05$, with an average MR of $0.15$, and shows an increasing trend in the average missclassification rate starting from $\beta=0.05$ to $\beta=1$. In this example, the local approach has brought an average improvement over the global approach of about $4.5$ percentage points.

\begin{figure}[h]
    \centering
    \includegraphics[width=0.7\linewidth]{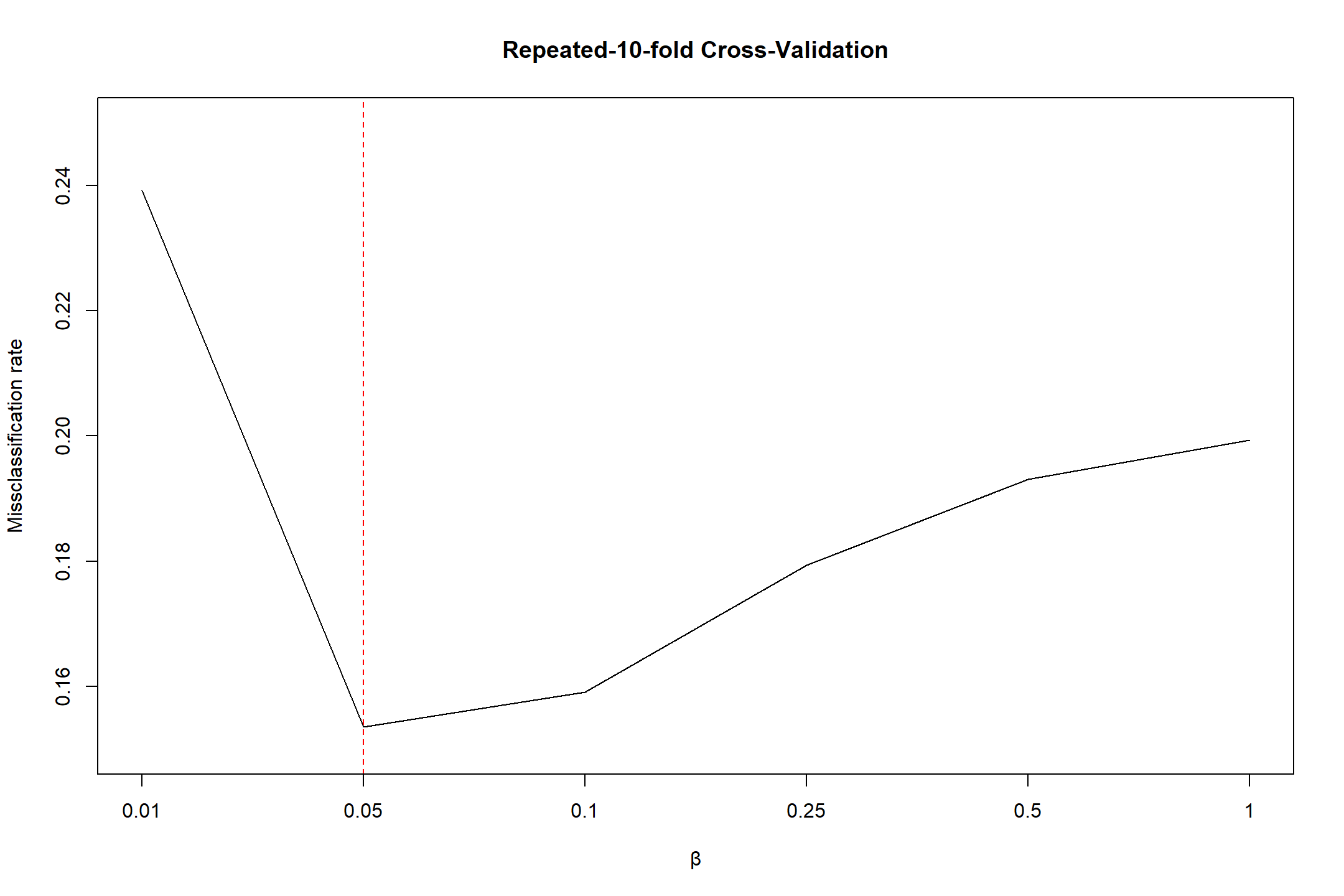}
    \caption{Repeated 10-fold cross-validation results for the Wholesales dataset. On the x--axis there are the $\beta $s, on the y--axis the cross--validated MR, and the dotted red line highlights the value of $\beta$ achieving the minimum MR.}
    \label{fig:whole}
\end{figure}

\subsection{SPAM database}
\label{sec:spam}

 The second dataset we propose classifies 4601 emails as spam or non-spam. What interested us about this textual dataset was both the relevance of the application, since spam detection is an important issue, as spam e-mails can range from simply annoying to actually dangerous, and the complexity of the data, which is of high dimensionality. In fact, there are a total of 57 variables, of which we select the last one, which contains information about the classification, and the first 48, which instead contain information on the percentage of words in the e-mail that match a certain word. Since the percentages of the 48 words chosen to classify the email do not sum up to 1, we added one last variable, which is the complement of 1 of the sum of these percentages. This is done to correctly visualize the point on a hypersphere through the square root transformation, since normalizing the data without it would give a very different interpretation to the results. Treating this as compositional data allows us to overcome the biases driven by the different lengths of the emails, which is the usual strategy in text mining applications \citep{Dhillon2004ConceptDF}.

Again, Fig. \ref{fig:spam} paints a clear picture of the choice of $\beta$. Quite interestingly, something that did not occur during our simulations, the best performing $\beta$ is 0.01, with an average MR of $0.12$, possibly due to the presence of a higher number of total observations. The figure also shows a trend that increases steadily up to $\beta=0.5$, and an average difference between the chosen local and the global depth of $8$ percentage points.

\begin{figure}[h]
    \centering
    \includegraphics[width=0.7\linewidth]{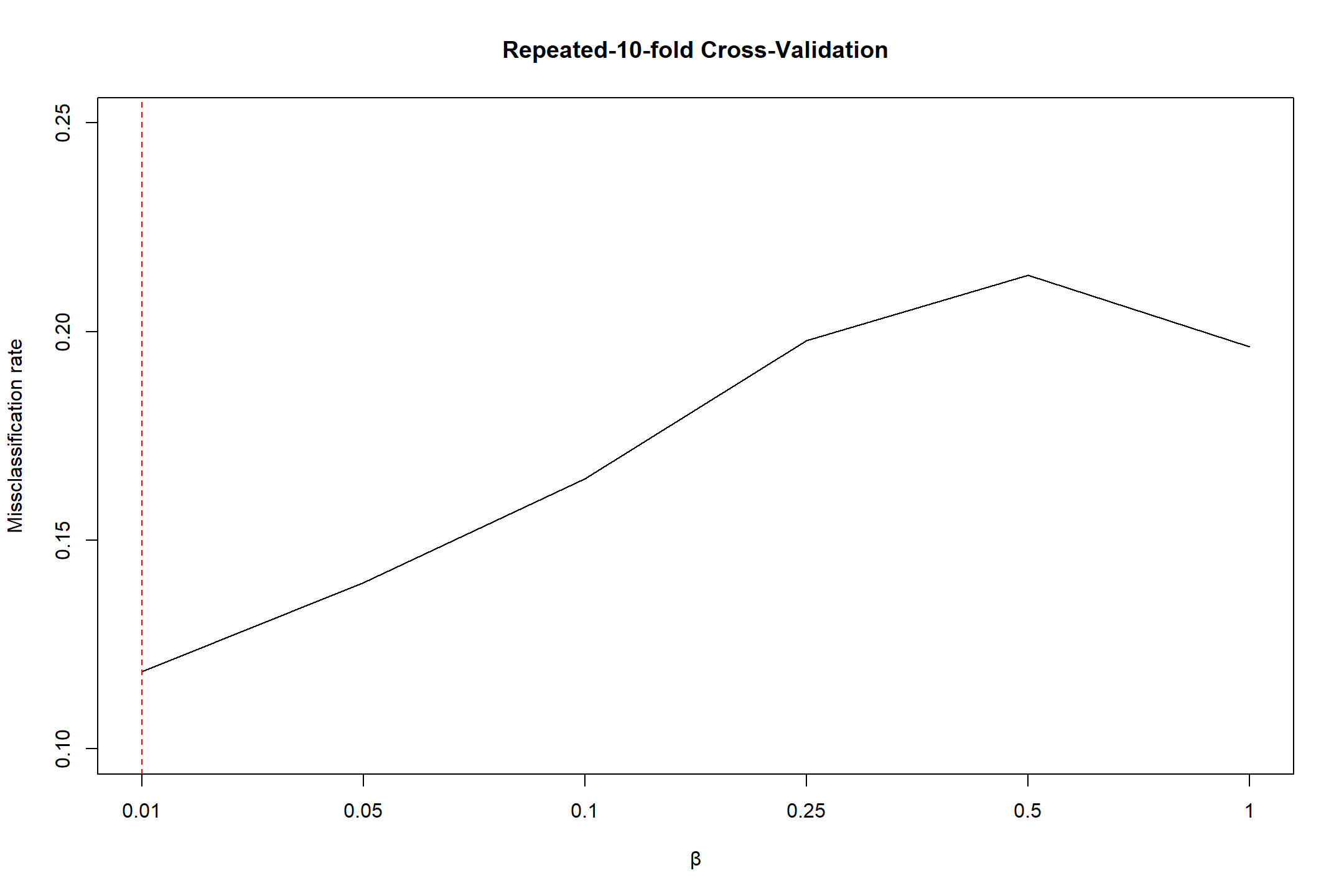}
    \caption{Repeated 10-fold cross-validation results for the Spam dataset. On the x--axis there are the $\beta $s, on the y--axis the cross--validated MR, and the dotted red line highlights the value of $\beta$ achieving the minimum MR.}
    \label{fig:spam}
\end{figure}

\section{Conclusion}
\label{concl}

This paper introduces the local cosine distance depth (LCDD) for directional data on the hypersphere. LCDD extends the cosine distance depth to capture local centrality in multimodal distributions using a neighbourhood approach. When applied to the Depth vs. Depth (DD) classifier, LCDD enables better separation, even for non-convex class structures.
The proposed LCDD-based classifier is compared with its global counterpart in an extensive simulation study. The results demonstrate the effectiveness of the local approach, regardless of the chosen $\beta$ level, except in cases of high noise and dimensions, where the global and local depth functions produce similar results.
These results are further confirmed by two real-data examples. 
Future research will focus on extending the approach to multiclass settings and different manifolds.

\bibliography{ref}

\end{document}